\documentclass[11pt]{article}

\usepackage[margin=1in]{geometry}
\usepackage{amsmath,amssymb,amsthm,mathtools}
\usepackage{enumitem}
\usepackage{booktabs}
\usepackage{graphicx}
\usepackage{float}
\usepackage{array,tabularx}
\usepackage{microtype}
\usepackage{natbib}
\bibpunct[, ]{(}{)}{,}{a}{}{,}
\usepackage[hidelinks]{hyperref}

\newtheorem{theorem}{Theorem}[section]
\newtheorem{proposition}[theorem]{Proposition}
\newtheorem{lemma}[theorem]{Lemma}
\newtheorem{corollary}[theorem]{Corollary}
\theoremstyle{definition}
\newtheorem{assumption}[theorem]{Assumption}

\theoremstyle{remark}
\newtheorem{remark}[theorem]{Remark}

\DeclareMathOperator{\Var}{Var}
\DeclareMathOperator{\dist}{dist}
\DeclareMathOperator{\sgn}{sgn}
\newcommand{\R}{\mathbb R}
\newcommand{\E}{\mathbb E}
\newcommand{\Prob}{\mathbb P}
\newcommand{\1}{\mathbf 1}
\newcommand{\dd}{\,\mathrm d}
\newcommand{\eps}{\varepsilon}
\newcommand{\OPT}{\mathrm{OPT}}
\newcommand{\ALG}{\mathrm{ALG}}
\newcommand{\Reg}{\mathrm{Regret}}

\newcommand{\calL}{\mathcal L}

\title{Tight Lower Bounds for the Multi-Secretary Problem via \\ Bellman Certificates}
\author{Jiawei Zhang\\Department of Technologies, Operations, and Statistics\\Stern School of Business, New York University\\\texttt{jz31@stern.nyu.edu}}
\date{\today}

\begin{document}
\maketitle

\begin{abstract}
This paper studies additive regret in the multi-secretary problem, defined as the gap between the expected offline prophet reward and the reward of the best online policy.  Prior work established $O(\log T)$ regret for bounded-density distributions with connected support and $O((\log T)^2)$ upper bounds for bounded-density distributions with support gaps.  It was unknown whether the extra logarithmic factor is necessary even in the one-resource model.  We prove that it is necessary.  For a mixture of two separated uniform distributions at the critical capacity, the optimal regret grows at least on the order of $(\log T)^2$.  Thus the existing $O((\log T)^2)$ upper bounds for bounded-density gapped instances, including those implied by network revenue management models with continuous rewards, are tight in this simplest specialization.  The same framework also yields a matching lower bound for gapped distributions whose gap-facing densities vanish near the support edges; this companion result is given in the appendix.  The proofs use Bellman certificates: feasible solutions to a relaxation of the exact Bellman recursion.  This framework converts lower bounds into explicit certificate constructions and identifies why support gaps permit larger regret.
\end{abstract}

\section{Introduction}\label{sec:intro}

In the multi-secretary problem, a decision maker observes independent values $V_1,\ldots,V_T\sim F$ sequentially and may accept at most $k$ of them.  The objective is to maximize the expected total accepted value.  We measure performance by additive regret against the offline, or prophet, benchmark: the expected reward of a decision maker who observes all $T$ realized values before selecting.  Because all values are nonnegative and the offline decision maker has only the cardinality constraint, the offline optimum is obtained by selecting the largest $k$ realized values.  Thus the prophet benchmark is the expected sum of the largest $k$ order statistics, and the regret is this benchmark minus the expected reward of the optimal online policy.

The order of this regret depends sharply on the local structure of the distribution $F$.  \citet{ArlottoGurvich2019} show that, when $F$ has finite support, the regret is uniformly bounded in the horizon.  For continuous distributions with connected support and density bounded above and away from zero on the relevant interval, such as a uniform distribution, \citet{Lueker1998} and \citet{Bray2024} show that the tight order is logarithmic.

\citet{BesbesKanoriaKumar2024} show that substantially larger regret can arise when the distribution places little probability mass near the relevant selection threshold. They organize such instances by a local mass-accumulation exponent $\beta$.  The case $\beta=0$ means that the density remains bounded away from zero near that threshold, and includes the uniform distribution.  The case $\beta>0$ means that the local probability mass in an interval of length $\varepsilon$ is of order $\varepsilon^{\beta+1}$, equivalently that the density behaves like distance$^\beta$ when a density exists.  For connected support and $\beta>0$, their upper bound has a leading polynomial factor $T^{\beta/(2(\beta+1))}$, and they prove a matching lower bound.

They also study distributions with separated support intervals.  On each side of the gap, the mass near the endpoint adjacent to the gap may again follow the exponent $\beta$. When $\beta=0$, the gap-facing densities are bounded above and away from zero, as in a mixture of two separated uniform distributions.  This bounded-density gapped class is also a single-resource special case of the network revenue management model with continuous reward distributions studied by \citet{JiangMaZhang2025Degeneracy}.  In this regime, \citet{BesbesKanoriaKumar2024} and \citet{JiangMaZhang2025Degeneracy} give log-squared upper bounds.  For $\beta>0$, \citet{BesbesKanoriaKumar2024} give an upper bound with the same polynomial term as in the connected-support case, but multiplied by an additional logarithmic factor; their lower bound has the polynomial term but lacks this logarithmic factor.

Our main results resolve the remaining lower-bound questions for the gapped regimes considered here.  First, for the bounded-density gapped case $\beta=0$, we prove an $\Omega((\log T)^2)$ lower bound for the simplest such instance: a mixture of two separated uniform distributions.  The lower bound holds at the critical capacity $k_T=\lfloor qT\rfloor$, where $q$ is the probability mass of the upper support.  Therefore the log-squared upper bound is tight already in the single-resource multi-secretary problem, and consequently the log-squared bound of \citet{JiangMaZhang2025Degeneracy} for bounded-density continuous-reward network revenue management cannot in general be improved.  Second, for gapped distributions with $\beta>0$, we prove a lower bound with the additional logarithmic factor appearing in the upper bound of \citet{BesbesKanoriaKumar2024}.  This positive-$\beta$ gapped lower bound is proved for capacities shifted from the critical capacity: $k_T$ is of the form $qT$ plus a fixed positive multiple of $\sqrt{T\log T}$, where $qT$ is the expected number of observations from the upper support.  Unlike the two-uniform bounded-density case, where the sharp lower bound is obtained at $\lfloor qT\rfloor$, the positive-$\beta$ gapped theorem is stated for this moderate-deviation shifted-capacity region.  Table~\ref{tab:intro-summary} summarizes the known results and the new lower bounds.  For compactness, the table writes
\[
        R_\beta(T):=T^{\beta/(2(\beta+1))},
        \qquad
        G_\beta(T):=R_\beta(T)(\log T)^{(\beta+2)/(2(\beta+1))}.
\]
\vspace{-2em}
\begin{table}[ht]
\centering
\caption{Known and new lower bounds in the four distributional regimes.}
\label{tab:intro-summary}
\scriptsize
\setlength{\tabcolsep}{3pt}
\begin{tabularx}{\textwidth}{>{\raggedright\arraybackslash}p{0.12\textwidth} >{\raggedright\arraybackslash}p{0.23\textwidth} >{\raggedright\arraybackslash}p{0.21\textwidth} >{\raggedright\arraybackslash}X}
\toprule
Case & Known upper bound & Known lower bound & New lower bound \\
\midrule
No gap, $\beta=0$
& $O(\log T)$ \citep{Bray2024}
& $\Omega(\log T)$ \citep{Lueker1998,Bray2024}
& -- \\
Gap, $\beta=0$
& $O((\log T)^2)$ \citep{BesbesKanoriaKumar2024,JiangMaZhang2025Degeneracy}
& $\Omega(\log T)$ \citep{Lueker1998,Bray2024}
& $\Omega((\log T)^2)$, Theorem~\ref{thm:gapped-zero} \\
No gap, $\beta>0$
& $O(R_\beta(T))$ \citep{BesbesKanoriaKumar2024}
& $\Omega(R_\beta(T))$ \citep{BesbesKanoriaKumar2024}
& -- \\
Gap, $\beta>0$
& $O(G_\beta(T))$ \citep{BesbesKanoriaKumar2024}
& $\Omega(R_\beta(T))$ \citep{BesbesKanoriaKumar2024}
& $\Omega(G_\beta(T))$, Theorem~\ref{thm:gapped-shift-positive} (Appendix~\ref{app:gapped-positive}) \\
\bottomrule
\end{tabularx}
\end{table}
\subsection{Lower bounds by Bellman certificates}\label{subsec:intro-certificates}

We prove these tight lower bounds by a unified Bellman-certificate method.  \citet{BesbesKanoriaKumar2024} prove their lower bounds through direct hard-instance constructions and indistinguishability arguments over arbitrary candidate online policies.  Our method is structurally different: we do not analyze any specific online policy, but instead construct explicit Bellman certificates, feasible solutions to relaxations of the exact Bellman recursion for for the gap between the offline benchmark and the optimal online value.

The online multi-secretary problem has a standard dynamic-programming formulation.  The offline prophet benchmark also admits a recursion based on order statistics: after separating one current observation from the future observations, the current observation improves the offline top-$c$ sum exactly when it exceeds the future marginal order statistic.  Subtracting the online and offline recursions yields a Bellman equation for the prophet-online regret.

We denote a candidate regret certificate by $B=(B_s(c))$, where $s$ is the number of remaining observations and $c$ is the remaining capacity.  The target capacity is typically close to $qs$, where $q$ is the probability mass of the upper support in the gapped model.  In the gapped two-uniform case, $q$ is the probability mass of the upper support, so $qs$ is the expected number of future observations from that upper support.  Let $H_{s,c}$ be the $c$-th largest value among those $s-1$ future observations.  Let $\tau_s(c)=\E H_{s,c}$ and let $D_s(c)=B_{s-1}(c)-B_{s-1}(c-1)$ be the one-step capacity first difference of the certificate.  The Bellman residual can be written as
\begin{equation}\label{eq:intro-bellman-residual}
\resizebox{\textwidth}{!}{$
\underbrace{(1-q)B_{s-1}(c)+qB_{s-1}(c-1)-B_s(c)}_{M_s(c):\ \text{deterministic drift term}}
+\underbrace{\{\E\Delta(H_{s,c})-\Delta(\tau_s(c))\}}_{\text{order-statistic Jensen term}}
+\underbrace{\{\Delta(\tau_s(c))-\Delta(\tau_s(c)-D_s(c))\}}_{\text{finite-difference perturbation}}
\ge 0.
$}
\end{equation}
Here $\Delta$ is the convex residual; its curvature depends on the local shape of the distribution $F$.  When $F$ has density bounded away from zero near the operating threshold, $\Delta$ is strongly convex locally.  When $F$ has a support gap, $\Delta$ is flat on the empty interval.

The optimal prophet-online regret satisfies the corresponding Bellman equation with equality.  We prove that, in order to derive a regret lower bound, it is enough to construct a nonnegative candidate $B$ satisfying the boundary and base conditions, together with the one-sided Bellman certificate inequality \eqref{eq:intro-bellman-residual} for all states in the finite-horizon grid, with the boundary cases $c=0$ and $c\ge s$ treated by the imposed zero boundary values.  The resulting comparison principle is the main lower-bound tool.

The certificate construction has the following interpretation.  A candidate lower bound can be made large at the target state, where $c$ is close to $qs$, only if it can be propagated through the Bellman inequality at every nearby state.  The term $M_s(c)$ measures the resulting drift cost: it records how much the proposed certificate decreases, after the natural affine transport of capacity, when one moves from time $s-1$ to time $s$.  A taller certificate creates a larger drift burden, governed by how quickly the profile changes across both capacity and time.  Feasibility requires this cost to be paid for by the remaining terms in \eqref{eq:intro-bellman-residual}.

The first remaining term is the order-statistic Jense term.  It is independent of the certificate $B$ and depends only on the distribution of the prophet's offline marginal item ($H_{s,c}$) and on the shape of $\Delta$.  The final term is the finite-difference perturbation.  It depends on the spatial first difference $D_s(c)$ of the certificate and can be negative.  Thus, in constructing a certificate, one must balance three effects: deterministic drift, positive order-statistic slack, and possible loss from the finite difference.

This decomposition isolates the support-gap mechanism.  The operating threshold lies in the empty interval between two support components.  In the moderate band where $c-qs$ is of order $\sqrt{s\log s}$, the certificate is constructed so that both the mean prophet threshold $\tau_s(c)$ and the shifted threshold $\tau_s(c)-D_s(c)$ remain inside this empty interval.  Consequently the finite-difference perturbation in \eqref{eq:intro-bellman-residual} is zero there.  Randomness still matters, however: on sample paths with slightly too many or slightly too few observations from the upper support, the prophet's marginal item can fall near one of the two support boundaries adjacent to the gap.  At those boundaries $\Delta$ starts to bend, and the order-statistic Jensen term is positive.  The support gap therefore separates where slack is generated from where the certificate is transported.  Fluctuations of the offline marginal into the support edges generate positive slack, while the flat gap suppresses the local finite-difference penalty.  This separation permits a larger feasible certificate and yields the order $(\log T)^2$ lower bound at $k_T=\lfloor qT\rfloor$ for the two-uniform bounded-density gapped case.

\subsection{Additional related literature}\label{subsec:intro-related}

The multi-secretary problem is a single-resource special case of broader online resource allocation, network revenue management, and online linear programming models; see \citet{BalseiroBesbesPizarro2024}. Classical fluid and deterministic-linear-program approaches to network revenue management date back at least to \citet{GallegoVanRyzin1994} and \citet{TalluriVanRyzin1998}, and static fluid or bid-price controls generally incur regret of order $O(\sqrt{T})$; much of the subsequent literature asks when this baseline can be improved.

When rewards and resource consumptions have finite support, \citet{JasinKumar2012} show that resolving the fluid relaxation obtains constant regret under standard nondegeneracy assumptions. Recent works obtain constant regret without the standard nondegeneracy assumption, including \citet{ArlottoGurvich2019}, \citet{BumpensantiWang2020}, \citet{VeraBanerjee2021}, \citet{VeraEtAl2021}, \citet{FreundZhao2022}, and \citet{BanerjeeFreund2024}. These algorithms typically require repeatedly resolving the fluid relaxation, with the number of resolves ranging from $O(T)$ to $O(\log\log T)$. More recently, \citet{Gupta2024} and \citet{HeWeiXuYu2025} present algorithms that solve the fluid relaxation only once; they prove constant-regret bounds under conditions equivalent to existing nondegeneracy assumptions. Several recent papers present constant-regret algorithms without assuming known probability distributions; see \citep{ChenLiYe2024,WeiXuYu2023,XieMaXin2025} for nondegenerate problems and \citet{LiWangZhang2024} for degenerate problems.

When rewards or consumptions are continuously distributed, constant regret is generally no longer attainable. For online linear programming with continuous valuations, \citet{GaoGeXueSunYe2025} show that first-order algorithms can achieve $o(\sqrt{T})$ regret. \citet{LiYe2022} prove polylogarithmic regret under local strong-convexity and smoothness conditions, and \citet{Bray2024} sharpens this to a tight logarithmic bound with a different set of assumptions. \citet{JiangZhang2020StochasticConsumption} extend the result of \citet{ArlottoXie2020} to obtain logarithmic regret under regularity conditions. \citet{ChenWang2025} study general continuous reward distributions without the standard primal-stability, strict-complementarity, or second-order growth assumptions, though their conditions still imply uniqueness of the fluid dual.

Most directly related to our work, \citet{JiangMaZhang2025Degeneracy} study network revenue management with finitely many resource-consumption vectors but continuous, bounded-density reward distributions, obtaining $O((\log T)^2)$ regret without any additional assumptions. Their model contains our bounded-density gapped multi-secretary instances as a single-resource special case, and our lower bound shows the log-squared rate is already tight there.

Relative to these works, our contribution is a lower-bound method for additive regret.  The Bellman-certificate framework turns the lower-bound problem into the construction of explicit feasible solutions to relaxations of the exact Bellman recursion.  This Bellman-inequality structure is closest in spirit to \citet{VeraEtAl2021}, who use Bellman inequalities to prove constant-regret guarantees for tractable online allocation and pricing policies.  Their inequalities certify that a particular online policy is close to the offline benchmark, yielding an upper bound on regret.  Here the certificate has the opposite role: the certificates are feasible solutions to relaxations of the exact Bellman recursion, and the comparison principle converts them into lower bounds on the gap incurred by every online policy.  The same framework applies to bounded-density support gaps and gap-facing vanishing densities.

The paper is also related to the prophet-inequality literature, which studies multiplicative competitive ratios.  Classical single-unit prophet inequalities compare an online stopping rule with the maximum realized value; see \citet{HillKertz1982}, \citet{SamuelCahn1984}, and \citet{CorreaEtAl2017Posted}.  A large subsequent literature studies prophet inequalities under cardinality, matroid, and other feasibility constraints; see \citet{ChawlaHartlineMalecSivan2010}, \citet{Alaei2014}, \citet{JiangMaZhang2024Prophet}, and \citet{JiangMaZhang2025Tightness}, and the survey of \citet{CorreaEtAl2019Survey}.

\subsection{Organization}\label{subsec:intro-organization}

Section~\ref{sec:model} defines the model, the gapped support structure, and the affine-reference residual function. Section~\ref{sec:framework} formulates the Bellman-certificate optimization problem and proves the comparison principle. Section~\ref{sec:gapped-zero} proves the two-uniform bounded-density lower bound using the certificate in \eqref{eq:g0-certificate}, with the moderate-scale finite-difference expansions deferred to Appendix~\ref{app:g0-asymp}. Appendices~\ref{app:binomial}--\ref{app:fd-standard} collect the reusable binomial, order-statistic, and finite-difference estimates. Appendix~\ref{app:gapped-positive} proves the gapped vanishing-density lower bound for $\beta>0$ using the framework of Section~\ref{sec:framework}.

\section{Model setup}\label{sec:model}

This section introduce the basic model and the gapped distributional structure used in the main body of the paper.  The main text focuses on bounded-density gapped distributions, with the two-uniform mixture as the sharp lower-bound instance.  Gapped distributions whose gap-facing densities vanish, as in \citet{BesbesKanoriaKumar2024}, are treated in Appendix~\ref{app:gapped-positive}.

Let $V_1,V_2,\ldots$ be i.i.d. draws from a distribution $F$ on $[0,\infty)$.  Throughout the paper $F$ is atomless and has compact support.  The compactness assumption is used only to avoid irrelevant integrability issues and to keep all value functions finite.

A policy with horizon $T$ and capacity $k$ observes $V_t$ before deciding whether to accept it.  Formally, a policy is a sequence $A_t\in\{0,1\}$ such that $A_t$ is measurable with respect to the history generated by $(V_1,\ldots,V_t,A_1,\ldots,A_{t-1})$, and $\sum_{t=1}^T A_t\le k$ almost surely.
The online reward is $\ALG_T^\pi(k):=\sum_{t=1}^T A_tV_t$.
The offline prophet reward is the sum of the largest $k$ realized values, $\OPT_T(k):=\sum_{i=1}^{k} V_{i:T}^{\downarrow}$, where $V_{1:T}^{\downarrow}\ge\cdots\ge V_{T:T}^{\downarrow}$ are the descending order statistics, and the convention is that the sum is over $i\le \min\{k,T\}$.  Since the support is nonnegative, the at-most-$k$ offline benchmark agrees with the sum of the largest $k$ observations.

In this paper, we focus on the gapped model, i.e., $F$ has two separated support intervals $I_-:=[a_1,b_1]$ and $I_+:=[a_2,b_2]$, where $0\le a_1<b_1<a_2<b_2$, with gap width $G:=a_2-b_1>0$.  The distribution is a mixture $F=pF_-+qF_+$, where $p:=1-q$, $F_-$ is supported on $I_-$, and $F_+$ is supported on $I_+$.  Equivalently, $q=\Prob(V\in I_+)$ is the upper-support mass. 

The main body focuses on bounded-density gapped distributions, and the lower bound is proved for the two-uniform mixture, where $F_-$ and $F_+$ are uniform on their respective support intervals.  The vanishing-density case, where the gap-facing mass grows at a higher-order power, is discussed in the appendix.

For $s\ge0$ and integer $c$, let $J_s(c):=\sup_{\pi}\E[\sum_{t=1}^s A_tV_t]$ and $\Phi_s(c):=\E[\sum_{i=1}^{\min\{c,s\}}V_{i:s}^{\downarrow}]$ be the optimal online and expected prophet values, respectively, with $s$ periods remaining and capacity $c$.
We use the boundary convention
\begin{equation}\label{eq:boundary-model}
        J_s(0)=\Phi_s(0)=0,
        \qquad
        J_s(c)=\Phi_s(c)=s\E V \quad \text{for }c\ge s.
\end{equation}
The prophet-online gap is
\[
        B_s^\star(c):=\Phi_s(c)-J_s(c).
\]
The optimal additive regret for horizon $T$ and capacity $k$ is
\[
        \Reg(T,k;F):=B_T^\star(k).
        \]

Throughout the paper, constants denoted by $c,C,C_1,\ldots$ may depend on the fixed distributional primitives, on $q$, and on fixed parameters in the constructed certificates, but never on $s,T,c,k$.  Constants in the appendix may also depend on the local edge exponent introduced there.  We write $A_s\lesssim B_s$ if $A_s\le C B_s$ for a constant $C<\infty$, and $A_s\asymp B_s$ if both $A_s\lesssim B_s$ and $B_s\lesssim A_s$.  All statements involving ``sufficiently large'' $s$ or $T$ are uniform over the state ranges explicitly specified in the corresponding lemma.

\section{Common Bellman-certificate framework}\label{sec:framework}

This section contains the dynamic-programming identities and the comparison principle used in all lower-bound constructions.  The arguments are independent of the particular certificate constructed later.

\subsection{Online and offline dynamic recursions}\label{subsec:framework-recursions}

We first present the Bellman recursion for the online problem.  The boundary conditions are those in \eqref{eq:boundary-model}.  Define the tail-integral function
\[
        h(\tau):=\E[(V-\tau)^+]
        =\int_\tau^\infty (1-F(v))\,\dd v.
\]
This function will be used throughout the paper.  Since $F$ is atomless, $h$ is continuously differentiable and $h'(\tau)=-(1-F(\tau))=F(\tau)-1$.  In particular, $h$ is convex because $h'$ is nondecreasing.
The following result is well-known in the literature. We include a proof for completeness.
\begin{proposition}[Online Bellman recursion]\label{prop:online-bellman-framework}
For $s\ge1$ and $1\le c<s$,
\begin{equation}\label{eq:online-recursion}
        J_s(c)
        =J_{s-1}(c)+h\!\left(J_{s-1}(c)-J_{s-1}(c-1)\right).
\end{equation}
The optimal online action at state $(s,c)$ is a threshold rule: accept the current value $V$ if and only if
\begin{equation}\label{eq:online-threshold-framework}
        V\ge J_{s-1}(c)-J_{s-1}(c-1).
\end{equation}
\end{proposition}

\begin{proof}
At state $(s,c)$, after observing $V$, rejection yields continuation value $J_{s-1}(c)$, while acceptance yields $V+J_{s-1}(c-1)$.  Therefore
\[
        J_s(c)=\E\max\{J_{s-1}(c),V+J_{s-1}(c-1)\}.
\]
Using $\max\{a,V+b\}=a+(V-(a-b))^+$ with $a=J_{s-1}(c)$, $b=J_{s-1}(c-1)$, and the definition of $h$ gives \eqref{eq:online-recursion}.  The same maximization gives the threshold rule \eqref{eq:online-threshold-framework}.  The state space is finite for each horizon, so the dynamic program attains the optimum.
\end{proof}

For the prophet recursion we separate one observation, called the current observation, from the remaining $s-1$ future observations.  Let $H_{s,c}$ be the $c$-th largest value among those $s-1$ future observations for $s\ge2$ and $1\le c<s$. We also refer to  $H_{s,c}$ as the offline marginal. Equivalently, if $S_r$ denotes the sum of the largest $r$ among the future observations, with $S_0=0$, then $H_{s,c}=S_c-S_{c-1}$.  Since $F$ is atomless, ties occur with probability zero; we use the standard order-statistic convention.
For $s\ge2$ and $1\le c<s$, define
\[
        \tau_s(c):=\Phi_{s-1}(c)-\Phi_{s-1}(c-1).
\]

\begin{proposition}[Offline marginal and prophet recursion]\label{prop:offline-bellman-framework}
It holds that 
\begin{equation}\label{eq:offline-recursion-general}
  \tau_s(c)=\E H_{s,c} \quad \mbox{and}   \quad   \Phi_s(c)=\Phi_{s-1}(c)+\E h(H_{s,c}).
\end{equation}
\end{proposition}

\begin{proof}
The first equality follows from $H_{s,c}=S_c-S_{c-1}$: by definition $\E S_c=\Phi_{s-1}(c)$ and $\E S_{c-1}=\Phi_{s-1}(c-1)$, hence
\[
        \E H_{s,c}=\E(S_c-S_{c-1})=\Phi_{s-1}(c)-\Phi_{s-1}(c-1).
\]
Now add an independent current observation $V$ to the future multiset.  The sum of the top $c$ values among all $s$ observations is
\[
        S_c+(V-H_{s,c})^+.
\]
The current value improves the previous top-$c$ block exactly when it exceeds the smallest value in that block, namely $H_{s,c}$; the improvement is then $V-H_{s,c}$.  Taking expectations and conditioning on $H_{s,c}$ gives
\[
        \Phi_s(c)=\E S_c+\E[(V-H_{s,c})^+]
        =\Phi_{s-1}(c)+\E h(H_{s,c}),
\]
which proves the second equality of \eqref{eq:offline-recursion-general}.
\end{proof}

For any $s\ge2$ and $1\le c<s$, define the prophet's one-step slack by
\[
        \delta_s(c):=\Phi_s(c)-\Phi_{s-1}(c)-h\bigl(\tau_s(c)\bigr).
        \]
On the gap $[b_1,a_2]$ we have $F(\tau)=p$, hence $h'(\tau)=F(\tau)-1=-q$ for $\tau\in(b_1,a_2)$.
Thus $h$ is affine on the gap.  Let $h_0$ denote its affine continuation from the gap to all of $\R$:
\[
        h_0(\tau):=h(b_1)-q(\tau-b_1).
\]
Thus
\begin{equation}\label{eq:h0-slope-model}
        h_0(\tau')-h_0(\tau)=-q(\tau'-\tau),
        \qquad \tau,\tau'\in\R.
\end{equation}
Define
\begin{equation}
        \Delta(\tau):=h(\tau)-h_0(\tau).
\end{equation}
Then $\Delta(\tau)=0$ for $\tau\in[b_1,a_2]$.
Since $h$ is convex and $h_0$ is affine, $\Delta$ is convex.  Moreover, $h_0$ agrees with a supporting affine segment of $h$ on the gap, so $\Delta$ is nonnegative on the support interval. 

\begin{lemma}[Offline slack identity]\label{lem:offline-slack-common}
For $s\ge2$ and $1\le c<s$,
\begin{equation}\label{eq:slack-common}
        \delta_s(c)=\E\Delta(H_{s,c})-\Delta\bigl(\tau_s(c)\bigr) \ge 0.
\end{equation}
\end{lemma}

\begin{proof}
By Proposition~\ref{prop:offline-bellman-framework},
\[
        \delta_s(c)=\E h(H_{s,c})-h(\E H_{s,c}).
\]
Write $h=h_0+\Delta$.  Since $h_0$ is affine,
\[
        \E h_0(H_{s,c})=h_0(\E H_{s,c}).
\]
The affine contribution therefore cancels, leaving \eqref{eq:slack-common}.  Since $\Delta$ is convex, Jensen's inequality gives
\[
        \E\Delta(H_{s,c})\ge \Delta(\E H_{s,c})=\Delta(\tau_s(c)).
\]
\end{proof}

\subsection{Bellman certificate}\label{subsec:framework-comparison}

For any candidate function $B=(B_s(c))$,  $s\ge2$, and $1\le c<s$, define the one-step difference $D_s(c)$, affine transport drift $M_s(c)$, and support-source term $S_s(c,D)$ by
\begin{align}
        D_s(c)&:=B_{s-1}(c)-B_{s-1}(c-1),\label{eq:D-common}\\
        M_s(c)&:=(1-q)B_{s-1}(c)+qB_{s-1}(c-1)-B_s(c),\\
        S_s(c,D)&:=\E\Delta(H_{s,c})-
        \Delta\bigl(\tau_s(c)-D\bigr).
        \label{eq:S-common}
\end{align}
By Lemma~\ref{lem:offline-slack-common}, $S_s(c,0)=\delta_s(c)$. With this notation, we now formulate a finite-dimensional optimization problem whose feasible solutions provide certified lower bounds on the prophet-online regret. Such feasible solutions are called Bellman certificates.

Fix a horizon $T$, a target capacity $k\in\{0,\ldots,T\}$, and a base time $s_0\ge2$.  
Consider the following maximization problem:
\begin{equation*}
\begin{array}{ll}
\text{maximize}_{B} & B_T(k) \\[2mm]
\text{subject to}
& B_s(c)\ge0,
        \qquad s_0\le s\le T,\;0\le c\le \min\{s,k\},\\[1mm]
& B_s(c)=0,
        \qquad s_0\le s\le T,\quad c=0 \ \mbox{or}\ c=s\le k,\\[1mm]
& B_{s_0}(c)\le B_{s_0}^\star(c),
        \qquad 0\le c\le \min\{s_0,k\}, \\[1mm]
& M_s(c)+S_s(c,D_s(c))\ge0,
        \qquad s_0<s\le T,
        \quad 1\le c \le \min\{s-1,k\},
\end{array}
\tag{$\mathsf P(T,k,s_0)$}\label{eq:certificate-program}
\end{equation*}
\addtocounter{equation}{1}
where $D_s$, $M_s$, and $S_s$ are defined by \eqref{eq:D-common}--\eqref{eq:S-common}.  We call the final constraint the Bellman certificate inequality; in pointwise form it is
\begin{equation}\label{eq:common-bellman-residual}
        M_s(c)+S_s(c,D_s(c))\ge0.
\end{equation}
A collection $B$ satisfying all constraints in \eqref{eq:certificate-program} is called a feasible Bellman certificate on $[s_0,T]$: it is a feasible solution to a relaxation of the exact Bellman recursion.

The program \eqref{eq:certificate-program} is not an algorithm for computing the regret.  Its role is dual: any feasible value is a certified regret lower bound.  The next proposition gives the exact Bellman-residual equation for the true gap and the resulting comparison principle.  It also shows that the optimal value of \eqref{eq:certificate-program} is exactly the true regret gap, because $B^\star$ itself is feasible with equality in the Bellman certificate constraints.

\begin{proposition}[Comparison principle]\label{prop:common-comparison}
Every feasible solution of \eqref{eq:certificate-program} satisfies
\begin{equation}\label{eq:certificate-lower-bound-framework}
        B_T(k)\le \Reg(T,k;F).
\end{equation}
The optimal value of \eqref{eq:certificate-program} is $B_T^\star(k)$.
\end{proposition}

\begin{proof}
We first record the exact residual identity satisfied by $B^\star$:
\begin{equation}\label {eq:true-bellman-residual-framework}
        M_s^\star(c)+S_s\bigl(c,D_s^\star(c)\bigr)=0.
\end{equation} Using the online recursion \eqref{eq:online-recursion}, the prophet recursion \eqref{eq:offline-recursion-general}, and the definition of $\delta_s(c)$,
$$    B_s^\star(c)
    =\Phi_s(c)-J_s(c)=\Phi_{s-1}(c)+h(\tau_s(c))+\delta_s(c)
          -J_{s-1}(c)
          -h\!\left(J_{s-1}(c)-J_{s-1}(c-1)\right).
$$
Since $J_{s-1}=\Phi_{s-1}-B_{s-1}^\star$, we have
\[
        J_{s-1}(c)-J_{s-1}(c-1)
        =\tau_s(c)-D_s^\star(c).
\]
Thus
\begin{equation}\label{eq:Bstar-h-form-framework}
        B_s^\star(c)
        =B_{s-1}^\star(c)+\delta_s(c)
        +h(\tau_s(c))
        -h\bigl(\tau_s(c)-D_s^\star(c)\bigr).
\end{equation}
By \eqref{eq:h0-slope-model}, for any $\tau$ and $D$,
\begin{equation}\label{eq:h-affine-split-framework}
        h(\tau)-h(\tau-D)
        =-qD+\Delta(\tau)-\Delta(\tau-D).
\end{equation}
Combining \eqref{eq:Bstar-h-form-framework}, \eqref{eq:h-affine-split-framework}, and \eqref{eq:slack-common} gives
\begin{align*}
        B_s^\star(c)
        &=B_{s-1}^\star(c)-qD_s^\star(c)
          +\E\Delta(H_{s,c})
          -\Delta\bigl(\tau_s(c)-D_s^\star(c)\bigr)\\
        &=(1-q)B_{s-1}^\star(c)+qB_{s-1}^\star(c-1)
          +S_s\bigl(c,D_s^\star(c)\bigr).
\end{align*}
Rearranging gives \eqref{eq:true-bellman-residual-framework}.

We next prove the comparison inequality by induction on $s$.  The base case $s=s_0$ is the base constraint in \eqref{eq:certificate-program}.  The boundary states $c=0$ and $c\ge s$ are immediate because both the online and prophet values agree there, and the program imposes the same zero boundary values on $B$.

Fix $s>s_0$ and assume $B_{s-1}(\cdot)\le B_{s-1}^\star(\cdot)$.  Let $1\le c<s$, and define, for real $x,y$,
\[
        \Psi_{s,c}(x,y):=(1-q)x+qy+S_s(c,x-y).
\]
By \eqref{eq:true-bellman-residual-framework},
\[
        B_s^\star(c)
        =\Psi_{s,c}\bigl(B_{s-1}^\star(c),B_{s-1}^\star(c-1)\bigr).
\]
The Bellman certificate constraint in \eqref{eq:certificate-program} is exactly
\[
        B_s(c)
        \le \Psi_{s,c}\bigl(B_{s-1}(c),B_{s-1}(c-1)\bigr).
\]
It remains to use monotonicity.  Set $\tau:=\tau_s(c)$ and $\delta:=\delta_s(c)$.  The map $\Psi_{s,c}$ admits the representation
\begin{equation}\label{eq:Psi-max-representation-framework}
        \Psi_{s,c}(x,y)
        =\Phi_s(c)-
        \E\max\left\{
              \Phi_{s-1}(c)-x,
              V+\Phi_{s-1}(c-1)-y
        \right\}.
\end{equation}
Apply $\E\max\{A,V+B\}=A+h(A-B)$ with
$A=\Phi_{s-1}(c)-x$ and $B=\Phi_{s-1}(c-1)-y$, and use
$\Phi_s(c)=\Phi_{s-1}(c)+h(\tau)+\delta$.  Increasing either $x$ or $y$ decreases one argument of the maximum in \eqref{eq:Psi-max-representation-framework}; hence $\Psi_{s,c}$ is nondecreasing in each coordinate.  Therefore
\[
        B_s(c)
   \le \Psi_{s,c}\bigl(B_{s-1}(c),B_{s-1}(c-1)\bigr)
   \le \Psi_{s,c}\bigl(B_{s-1}^\star(c),B_{s-1}^\star(c-1)\bigr)
        =B_s^\star(c),
\]
where the middle inequality uses the induction hypothesis.  Inequality \eqref{eq:certificate-lower-bound-framework} follows from $B_T^\star(k)=\Reg(T,k;F)$.

Finally, $B^\star$ itself is feasible: nonnegativity follows because the prophet dominates any online policy, the boundary and base constraints hold with equality, and the Bellman certificate constraints hold with equality by \eqref{eq:true-bellman-residual-framework}.  Hence the objective value $B_T^\star(k)$ is attainable in \eqref{eq:certificate-program}.  Together with the comparison part, this proves that the program value is exactly $B_T^\star(k)$.
\end{proof}

\subsection{The finite base condition}\label{subsec:framework-base}

The explicit certificates used later are asymptotic constructions valid for all sufficiently large $s$.  The comparison principle starts at a finite base time $s_0$, so we need a simple way to verify the base condition after fixing $s_0$.

\begin{lemma}[Strict finite-time prophet advantage]\label{lem:strict-common}
Assume $F$ is atomless and nondegenerate.  Then, for every $s\ge2$ and every $1\le c<s$, it holds that
$B_s^\star(c)>0.$
\end{lemma}

\begin{proof}
Let $\pi$ be any nonanticipating policy with capacity $c$, and let
\[
        \ALG^\pi:=\sum_{t=1}^s A_tV_t,
        \qquad
        \OPT:=\sum_{i=1}^c V_{i:s}^{\downarrow}.
\]
Pathwise, every feasible online policy satisfies $\ALG^\pi\le\OPT$.  It therefore suffices to show that no nonanticipating policy can satisfy $\ALG^\pi=\OPT$ almost surely.

Suppose, toward a contradiction, that some policy satisfies $\ALG^\pi=\OPT$ almost surely.  Since $F$ is atomless, the $s$ realized values are almost surely distinct.  On this event, the unique subset attaining $\OPT$ is the set of the $c$ largest observations.  Thus the policy must select exactly this top-$c$ subset almost surely.  In particular,
\begin{equation}\label{eq:A1-topc-framework}
        A_1=
        \1\{V_1\text{ belongs to the top-}c\text{ observations among }V_1,\ldots,V_s\}
        \quad\text{a.s.}
\end{equation}
Choose an interval $I$ with positive $F$-mass such that
\[
        0<\inf_{v\in I}\Prob(V>v)
        \le \sup_{v\in I}\Prob(V>v)<1;
\]
such an interval exists because $F$ is atomless and nondegenerate.  Conditional on $V_1\in I$, the future observations remain independent of the first decision and have a positive probability of producing at least $c$ values above $V_1$, and also a positive probability of producing fewer than $c$ values above $V_1$.  Equivalently, by the regular conditional distribution of $V_1$, for $F$-almost every $v\in I$ the binomial variable
\[
        N_v:=\#\{2\le t\le s:V_t>v\}\sim\mathrm{Bin}(s-1,\Prob(V>v))
\]
satisfies
\[
        \Prob(N_v<c)>0,
        \qquad
        \Prob(N_v\ge c)>0.
\]
For such $v$, the event that $V_1$ belongs to the top-$c$ block is exactly $\{N_v<c\}$, up to a null tie event.

The first decision $A_1$ is measurable with respect to the information available at time one, together with any policy randomization independent of the future; hence, under the regular conditional law given $V_1=v$, it is independent of $N_v$.  Let $a(v):=\Prob(A_1=1\mid V_1=v)$.  If \eqref{eq:A1-topc-framework} held almost surely, then for $F$-almost every such $v$ we would have $A_1=\1\{N_v<c\}$ almost surely under the conditional law.  Therefore
\[
        0=\Prob(A_1=1,N_v\ge c\mid V_1=v)=a(v)\Prob(N_v\ge c),
\]
which forces $a(v)=0$, and also
\[
        0=\Prob(A_1=0,N_v<c\mid V_1=v)=(1-a(v))\Prob(N_v<c),
\]
which forces $a(v)=1$.  This contradiction proves that every online policy loses to the prophet with positive probability.  Since an optimal online policy is attained by Proposition~\ref{prop:online-bellman-framework}, the expected gap $B_s^\star(c)$ is strictly positive.
\end{proof}

\begin{corollary}[Choosing the base scale]\label{cor:base-scale-framework}
Fix $s_0\ge2$.  Let $\widetilde B_s(c)=\eta\,\bar B_s(c)$ be a nonnegative candidate family depending linearly on a scale parameter $\eta>0$, and suppose $\bar B_{s_0}(c)<\infty$ for every $0\le c\le s_0$.  If the Bellman certificate inequality is verified for all $s>s_0$ whenever $0<\eta\le\eta_0$, then, after possibly shrinking $\eta\in(0,\eta_0]$, the base condition $ \widetilde B_{s_0}(c)\le B_{s_0}^\star(c)$ holds for $0 \le c\le s_0$.
\end{corollary}

\begin{proof}
The boundary states satisfy $\widetilde B_{s_0}(0)=\widetilde B_{s_0}(s_0)=B_{s_0}^\star(0)=B_{s_0}^\star(s_0)=0$ in the constructions below.  On the finite set of interior states, Lemma~\ref{lem:strict-common} gives
\[
        m_0:=\min_{1\le c<s_0}B_{s_0}^\star(c)>0.
\]
Let $K_0:=\max_{1\le c<s_0}\bar B_{s_0}(c)<\infty$.  If $K_0=0$ there is nothing to prove; otherwise choose $\eta\le m_0/K_0$, in addition to $\eta\le\eta_0$.
\end{proof}

\section[Gapped bounded-density case]{Mixture of Two Uniform Distributions}\label{sec:gapped-zero}

This section treats the bounded-density case in the gapped model, which corresponds to the case $\beta=0$ in the more general distribution introduced in Appendix~\ref{app:edge-order}.  We focus on the two-uniform mixture, which gives the sharp bounded-density lower bound.  Throughout this section we assume
\begin{equation}\label{eq:g0-uniform-model}
        F_-=\mathrm{Unif}[a_1,b_1],
        \qquad
        F_+=\mathrm{Unif}[a_2,b_2].
\end{equation}
Denote $L_-:=b_1-a_1$ and $L_+:=b_2-a_2$.
We now prove the main lower bound of the paper: the bounded-density gapped case has regret of order $(\log T)^2$.  

The proof builds a certificate on the moderate-deviation scale $\sqrt{s\log s}$, with height $(\log s)^2$ at the critical capacity.  The verification has four regimes, summarized below: the near-critical range, the bounded moderate-deviation range, the large moderate-deviation range, and the outer cutoff range.

A direct calculation for the two-uniform mixture gives
\begin{equation}\label{eq:g0-Delta-exact}
        \Delta(\tau)
        =\frac{q}{2L_+}(\tau-a_2)_+^2
        +\frac{p}{2L_-}(b_1-\tau)_+^2,
        \qquad \tau\in[a_1,b_2].
\end{equation}
This follows by integrating $h'(\tau)+q=F(\tau)-p$ on the lower and upper support intervals and using $\Delta=0$ on $[b_1,a_2]$.

\subsection{Candidate Certificate}
In this section, we construct a candidate certificate $B$ as follows.
For $s\ge s_0$ and $1\le c<s$, define
\begin{equation}\label{eq:g0-certificate}
        B_s(c):=\eta(\log s)^2\Omega(c/s)\varphi_s(z_s(c)),
\end{equation}
and set $B_s(0)=0$ and $B_s(c)=0$ for $c\ge s$.   Here $\eta$, $\Omega$, and $\varphi_{s}$ are constructed below.

  We choose constants in the following order.  First choose
\begin{equation}
        0<\alpha_0<\alpha_1<\alpha_2<\alpha_3<\alpha_4<\alpha_5<\frac12.
\end{equation}
Choose $\rho>0$ such that $\rho < \min\{p/4, q/4\}$. 
Choose constants $f_0>0$, $\kappa>8f_0$, and $\varphi_\star>f_0+\kappa\alpha_2^2$.  Finally choose $\lambda>0$ so large that
\begin{equation}\label{eq:g0-lambda-choice}
        \Bigl(\frac{11}{16}-\frac12\Bigr)\lambda\alpha_5>3.
\end{equation}
This lower bound on $\lambda$ is imposed for the large-deviation verification in Proposition~\ref{prop:g0-feasible}.  The small multiplicative parameter $\eta>0$ and the base time $s_0$ are chosen last.

Let $\Omega\in C^2([0,1])$ be a cutoff satisfying
\begin{equation}
        0\le\Omega\le1,
        \qquad
        \Omega(x)=1\text{ if } |x-q|\le\rho,
        \qquad
        \Omega(x)=0\text{ if } |x-q|\ge2\rho.
\end{equation}

Let $\psi:\mathbb R\to[0,1]$ be a $C^\infty$ transition function such that $\psi=0$ on $(-\infty,0]$, $\psi=1$ on $[1,\infty)$, $\psi'\ge0$, and all derivatives vanish at $0$ and $1$.  Define
\[
        \psi_1(z):=\psi\left(\frac{z-\alpha_1}{\alpha_2-\alpha_1}\right),
        \qquad
        \psi_2(z):=\psi\left(\frac{z-\alpha_3}{\alpha_4-\alpha_3}\right).
\]
Define
For all sufficiently large $s$, define $\varphi_s:[0,\infty)\to(0,\infty)$ by
\begin{equation}\label{eq:g0-profile-def}
\varphi_s(z):=
\begin{cases}
 f_0+\kappa z^2,
        &0\le z\le\alpha_1,\\
 (1-\psi_1(z))(f_0+\kappa z^2)+\psi_1(z)\varphi_\star,
        &\alpha_1\le z\le\alpha_2,\\
 \varphi_\star,
        &\alpha_2\le z\le\alpha_3,\\
 (1-\psi_2(z))\varphi_\star+\psi_2(z)\mathcal T_s(z),
        &\alpha_3\le z\le\alpha_4,\\
 \mathcal T_s(z):=\varphi_\star
        \exp\left\{-\frac{\lambda}{\log s}\left(z-\frac{\alpha_3+\alpha_4}{2}\right)\right\},
        &z\ge\alpha_4.
\end{cases}
\end{equation}
The flatness of $\psi$ at the joining points makes $\varphi_s$ a $C^3$ function.

\begin{lemma}[Profile bounds]\label{lem:g0-profile}
Uniformly in all sufficiently large $s$, the functions $\varphi_s$ defined in \eqref{eq:g0-profile-def} satisfy:
\begin{enumerate}[label=(P\arabic*),leftmargin=2.5em]
\item For $z \in [\alpha_1,\alpha_2]$, $\varphi_s'(z)\ge0$ and $|\varphi_s^{(j)}(z)|\le C_\lambda$ for $j=1,2,3$.
\item For $z \in [\alpha_3,\alpha_4]$, $\varphi_s(z)\asymp \varphi_\star$ and $
        |\varphi_s^{(j)}(z)|\le C_\lambda \frac{\varphi_s(z)}{\log s},
        \qquad j=1,2,3.
$
\item For every fixed $C_0<\infty$, uniformly over $\alpha_3\le z\le C_0\sqrt{s/(\log s)}$,
\begin{equation}\label{eq:g0-profile-time}
        |\varphi_{s-1}(z)-\varphi_s(z)|
        \le C_\lambda\frac{(1+z)\varphi_s(z)}{s(\log s)^2}.
\end{equation}
For every fixed $R<\infty$, uniformly over $0\le z\le R$ and $j=1,2$,
\begin{equation}\label{eq:g0-profile-der-time}
        |\varphi_{s-1}^{(j)}(z)-\varphi_s^{(j)}(z)|
        \le \frac{C_{\lambda,R}}{s(\log s)^2}.
\end{equation}
\end{enumerate}
\end{lemma}
\begin{proof}
On $[\alpha_1,\alpha_2]$,
\[
        \varphi_s'(z)=\psi_1'(z)\bigl(\varphi_\star-f_0-\kappa z^2\bigr)
        +(1-\psi_1(z))2\kappa z.
\]
Since $\varphi_\star>f_0+\kappa\alpha_2^2$, both terms are nonnegative; the derivative bounds on this fixed compact interval are immediate.  On $[\alpha_3,\alpha_4]$, the exponent in $\mathcal T_s$ is $O_\lambda(1/\log s)$, so $\mathcal T_s=\varphi_\star(1+O_\lambda(1/\log s))$.  More explicitly,
\[
        \mathcal T_s-\varphi_\star=O_\lambda(\varphi_\star/\log s),
        \qquad
        \mathcal T_s^{(m)}=(-\lambda/\log s)^m\mathcal T_s,
        \quad m\ge1.
\]
Because
\[
        \varphi_s=\varphi_\star+\psi_2(\mathcal T_s-\varphi_\star)
        \qquad (\alpha_3\le z\le\alpha_4),
\]
the first two derivatives are
\begin{align*}
        \varphi_s'
        &=\psi_2'(\mathcal T_s-\varphi_\star)+\psi_2\mathcal T_s',\\
        \varphi_s''
        &=\psi_2''(\mathcal T_s-\varphi_\star)
          +2\psi_2'\mathcal T_s'
          +\psi_2\mathcal T_s''.
\end{align*}
Each displayed term is $O_\lambda(\varphi_\star/\log s)$, and $\varphi_s\asymp\varphi_\star$ on $[\alpha_3,\alpha_4]$.  The third derivative is identical in form, with one more product-rule term, and is also $O_\lambda(\varphi_s/\log s)$.  This proves (P2).

It remains to prove the time-variation bounds.  On $[0,\alpha_3]$, the construction is independent of $s$.  On $[\alpha_3,\infty)$, view $\varphi_s(z)$ as $F(z,u)$ evaluated at $u=1/\log s$.  On $z\ge\alpha_4$,
\[
        \partial_u F(z,u)
        =-\lambda\left(z-\frac{\alpha_3+\alpha_4}{2}\right)F(z,u),
\]
and on $[\alpha_3,\alpha_4]$ the same expression is multiplied by the fixed cutoff $\psi_2(z)$.  Thus, uniformly for $z\ge\alpha_3$,
\[
        |\partial_u F(z,u)|
        \le C_\lambda(1+z)F(z,u).
\]
Moreover,
\[
        \frac1{\log(s-1)}-\frac1{\log s}
        =\frac{1}{s(\log s)^2}+O\left(\frac{1}{s(\log s)^3}\right).
\]
For $z\le C_0\sqrt{s/(\log s)}$ and $u$ between $1/\log s$ and $1/\log(s-1)$, the ratio
$F(z,u)/F(z,1/\log s)$ is bounded by a constant depending only on $C_0$ and $\lambda$, because
\[
        \left|u-\frac1{\log s}\right|(1+z)
        \le C\frac{1+z}{s(\log s)^2}=o(1).
\]
The mean-value theorem therefore gives \eqref{eq:g0-profile-time}.  For the derivative bounds on bounded intervals, differentiate $F$ first.  For instance, on $z\ge\alpha_4$,
\[
        \partial_u\partial_z F
        =-\lambda F
          +\lambda^2u\left(z-\frac{\alpha_3+\alpha_4}{2}\right)F,
\]
and similarly $|\partial_u\partial_z^2F|\le C_{\lambda,R}F$ on $z\le R$.  The interval $[\alpha_3,\alpha_4]$ has only additional fixed cutoff factors.  Applying the mean-value theorem to $\partial_zF$ and $\partial_z^2F$ gives \eqref{eq:g0-profile-der-time}.
\end{proof}

In order to verify the feasibility of the constructed function $B$, notice that Bellman certificate inequality \eqref{eq:common-bellman-residual} is equivalent to $$
\resizebox{\textwidth}{!}{$
\underbrace{(1-q)B_{s-1}(c)+qB_{s-1}(c-1)-B_s(c)}_{M_s(c):\ \text{deterministic drift}}
+\underbrace{\{\E\Delta(H_{s,c})-\Delta(\tau_s(c))\}}_{\text{Jensen/order-statistic source}}
+\underbrace{\{\Delta(\tau_s(c))-\Delta(\tau_s(c)-D_s(c))\}}_{\text{finite-difference perturbation}}
\ge 0.
$}
$$ We decompose the Bellman residual into three contributions: the deterministic
drift, the Jensen/order-statistic slack, and the finite-difference perturbation.
In the verification, the drift $M_s(c)$ is bounded separately.  The remaining
two contributions are controlled either separately, when the perturbation is
zero or negligible, or together through
\[
        S_s(c,D):=\E\Delta(H_{s,c})-\Delta(\tau_s(c)-D),
\]
when the perturbation is comparable to the Jensen/order-statistic slack.  Recall that 
$     D_s(c):=B_{s-1}(c)-B_{s-1}(c-1).$

Although the constructed certificates are verified on a broad state space, the most interesting cases are when $c$ is close to critical capacity $qs$, where $q$ is the upper-support mass introduced earlier.  Write $\sigma^2:=q(1-q)$.  For any $(s,c)$, define the centered capacity imbalance and its standard-deviation normalization by
\begin{equation}\label{eq:coords-model}
        d_s(c):=c-qs,
        \qquad
        x_s(c):=\frac{d_s(c)}{\sigma\sqrt{s}}.
\end{equation}
The proofs also use the moderate-deviation coordinate
\begin{equation}\label{eq:z-def-model}
        z_s(c):=\frac{|d_s(c)|}{\sigma\sqrt{s\log s}}
        =\frac{|x_s(c)|}{\sqrt{\log s}}.
\end{equation}
We begin with bounding $M_s(c)$, the deterministic drift in the next subsection. 

\subsection{Deterministic finite differences}
The following lemma bounds $M_s(c)$ and $D_s(c)$ in four regions, using the moderate-deviation coordinate $z_s(c)$.
\begin{lemma}[Bounded-density finite differences]\label{lem:g0-fdiff}
After increasing $s_0$ if necessary, the following estimates hold for all $s>s_0$ and all $1\le c<s$.
\begin{enumerate}[label=(D\arabic*),leftmargin=2.5em]
\item If $z_s(c)\le\alpha_0$ and $|c/s-q|\le\rho/2$, then
\begin{equation}
        M_s(c)\ge0,
        \qquad
        |D_s(c)|\le C\eta\frac{(\log s)^{3/2}}{\sqrt{s}}.
\end{equation}
\item If $\alpha_0\le z_s(c)\le\alpha_5$ and $|c/s-q|\le\rho/2$, then
\begin{equation}\label{eq:g0-annulus-MD}
        M_s(c)\ge -C_\lambda\eta\frac{\log s}{s},
        \qquad
        |D_s(c)|\le C_\lambda\eta\frac{(\log s)^{3/2}}{\sqrt{s}}.
\end{equation}
\item If $z_s(c)\ge\alpha_5$ and $|c/s-q|\le\rho/2$, then $d_s(c)D_s(c)<0$.  Moreover, uniformly for $\alpha_5\le z_s(c)\le C_\rho\sqrt{s/(\log s)}$, which is the range implied by $|c/s-q|\le\rho/2$,
\begin{align}
        |D_s(c)|
        &=(1+o(1))\frac{\eta\lambda}{\sigma}\varphi_s(z_s(c))\sqrt{\frac{\log s}{s}},
        \label{eq:g0-tail-D}\\
        M_s(c)
        &=-\eta \varphi_s(z_s(c))
        \left(2+\frac{\lambda z_s(c)}2+r_s(z_s(c))\right)\frac{\log s}{s}.
        \label{eq:g0-tail-M}
\end{align}
Here the remainder satisfies
\[
        \sup_{\alpha_5\le z\le C_\rho\sqrt{s/(\log s)}}
        \frac{|r_s(z)|}{1+\lambda z}\longrightarrow 0.
\]
\item If $|c/s-q|\ge\rho/2$ and at least one of $B_s(c)$, $B_{s-1}(c)$, $B_{s-1}(c-1)$ is nonzero, then
\begin{equation}\label{eq:g0-macro-MD}
        |M_s(c)|+|D_s(c)|=o(1/s).
\end{equation}
\end{enumerate}
\end{lemma}

\begin{proof}
When $|c/s-q|\le\rho/2$, the cutoff factors in $B_s(c)$, $B_{s-1}(c)$, and $B_{s-1}(c-1)$ are all equal to one for large $s$.  Write $d=d_s(c)$.  The predecessor deviations are
\[
        d_{s-1}(c)=d+q,
        \qquad
        d_{s-1}(c-1)=d-(1-q).
\]
Whenever $|d|\to\infty$ and the two predecessor deviations have the same sign as $d$, let $\xi$ be the two-point random variable taking values $q\,\sgn(d)$ with probability $1-q$ and $-(1-q)\,\sgn(d)$ with probability $q$.  Then $\E\xi=0$ and $\E\xi^2=\sigma^2$.  With this notation, Lemma~\ref{lem:g0-sign-stable-detail} gives the useful identity
\begin{equation}\label{eq:g0-fdiff-master-identity}
        M_s(c)
        =\eta\left\{(\log(s-1))^2\E\varphi_{s-1}(Z_s)
              -(\log s)^2\varphi_s(z)\right\},
        \qquad z=z_s(c),
\end{equation}
where $Z_s=\frac{|d|+\xi}{\sigma\sqrt{(s-1)\log(s-1)}}$.  On the same sign-stable range,
\begin{equation}
        D_s(c)
        =\eta(\log(s-1))^2
        \left[\varphi_{s-1}(z_{s-1}(c))
              -\varphi_{s-1}(z_{s-1}(c-1))\right].
\end{equation}
The range $z_s(c)\le\alpha_0$ is the only near-critical range where $d$ may be bounded; there we use the quadratic formula directly instead of \eqref{eq:g0-fdiff-master-identity}.

\emph{Regime 1.}  If $z_s(c)\le\alpha_0$, then for large $s$ both predecessor coordinates lie in the interval $[0,\alpha_1]$.  Therefore
\[
\varphi_{s-1}(z_{s-1}(c))=f_0+\kappa\frac{(d+q)^2}{\sigma^2(s-1)\log(s-1)},
\]
\[
\varphi_{s-1}(z_{s-1}(c-1))=f_0+\kappa\frac{(d-(1-q))^2}{\sigma^2(s-1)\log(s-1)}.
\]
Using
\begin{equation}
        (1-q)(d+q)^2+q(d-(1-q))^2=d^2+\sigma^2,
\end{equation}
we obtain
\[
\begin{aligned}
        M_s(c)
        &=\eta f_0((\log(s-1))^2-(\log s)^2)  \\
        &\quad+\eta\kappa\left[(\log(s-1))^2\frac{d^2+\sigma^2}{\sigma^2(s-1)\log(s-1)}
               -(\log s)^2\frac{d^2}{\sigma^2s\log s}\right].
\end{aligned}
\]
Since the square of the moderate-deviation normalizer is $\sigma^2s\log s$, this becomes
\[
        M_s(c)
        =\eta f_0((\log(s-1))^2-(\log s)^2)
        +\eta\kappa\left[\log(s-1)\frac{d^2+\sigma^2}{\sigma^2(s-1)}
               -\log s\frac{d^2}{\sigma^2s}\right].
\]
The $d^2$ part is
\[
        \frac{d^2}{\sigma^2}
        \left(\frac{\log(s-1)}{s-1}-\frac{\log s}{s}\right).
\]
It is nonnegative for large $s$ because $x\mapsto (\log x)/x$ is decreasing on $(e,\infty)$, so
\[
        \frac{\log(s-1)}{s-1}\ge \frac{\log s}{s}.
\]
The remaining positive contribution from the $\sigma^2$ term is
\[
        \eta\kappa\frac{\log(s-1)}{s-1}
        \ge \eta\kappa\frac{\log s}{2s}
\]
for all large $s$.  Finally,
\[
        (\log(s-1))^2-(\log s)^2
        =-\frac{2\log s}{s}+O\left(\frac{\log s}{s^2}\right),
\]
so the negative part contributed by the constant term $f_0$ is at most
$3\eta f_0\log s/s$ for large $s$.  Therefore
\[
        M_s(c)
        \ge \eta\frac{\log s}{s}\left(-3f_0+\frac{\kappa}{2}\right)
        \ge0
\]
for large $s$, because $\kappa>8f_0$.

The bound on $D_s(c)$ follows from the mean-value theorem and the global Lipschitz estimate
\begin{equation}\label{eq:g0-global-lipschitz-B}
        |\partial_c B_{s-1}(c)|
        \le C\eta\frac{(\log s)^2}{\sigma\sqrt{s\log s}}
        =C\eta\frac{(\log s)^{3/2}}{\sqrt{s}}.
\end{equation}
The estimate \eqref{eq:g0-global-lipschitz-B} follows by differentiating $(\log(s-1))^2\Omega(c/(s-1))\varphi_{s-1}(z_{s-1}(c))$.  The cutoff derivative contributes $O(\eta(\log s)^2/s)$, while the profile derivative contributes $O(\eta(\log s)^2/(\sigma\sqrt{s\log s}))$, which dominates.

\emph{Regime 2.}  Suppose $\alpha_0\le z_s(c)\le\alpha_5$.  Then the predecessor deviations have the same sign as $d$ for large $s$, and with the above definition of $\xi$,
\[
        (1-q)\varphi_{s-1}(z_{s-1}(c))+q\varphi_{s-1}(z_{s-1}(c-1))
        =\E \varphi_{s-1}\left(\frac{|d|+\xi}{\sigma\sqrt{(s-1)\log(s-1)}}\right).
\]
Let $z=z_s(c)$.  The random argument has expansion
\[
        \frac{|d|+\xi}{\sigma\sqrt{(s-1)\log(s-1)}}
        =z+\frac{\xi}{\sigma\sqrt{s\log s}}+\frac{z}{2s}+O\left(\frac{1+z}{s\log s}\right),
\]
so its first two centered moments around $z$ are
\[
        \E(Z-z)=\frac{z}{2s}+O\left(\frac{1+z}{s\log s}\right),
        \qquad
        \E(Z-z)^2=\frac{1}{s\log s}+O\left(\frac{1+z^2}{s^2}+\frac{1+z}{(s\log s)^{3/2}}\right).
\]
Decompose the drift using \eqref{eq:g0-fdiff-master-identity}:
\begin{equation}\label{eq:g0-annulus-decomp}
        M_s(c)/\eta
=\underbrace{(\log(s-1))^2
          \E[\varphi_{s-1}(Z_s)-\varphi_{s-1}(z)]}_{\text{coordinate transport}}
\quad+
          \underbrace{(\log(s-1))^2\varphi_{s-1}(z)-(\log s)^2\varphi_s(z)}_{\text{height/time change}}.
\end{equation}
The second term in \eqref{eq:g0-annulus-decomp} is
\begin{eqnarray*}
        &&(\log(s-1))^2\varphi_{s-1}(z)-(\log s)^2\varphi_s(z)\\
        &=&((\log(s-1))^2-(\log s)^2)\varphi_s(z)
          +(\log(s-1))^2(\varphi_{s-1}(z)-\varphi_s(z))\\
        &=&-2\varphi_s(z)\frac{\log s}{s}
          +O_\lambda\left(\frac1s\right).
\end{eqnarray*}
Here the first equality is exact, the scalar expansion
$(\log(s-1))^2-(\log s)^2=-2\log s/s+O(\log s/s^2)$ gives the main term.  Up to $\alpha_3$ the profile is independent of $s$, while on $[\alpha_3,\alpha_5]$ the time-variation estimate \eqref{eq:g0-profile-time} contributes only $O_\lambda(1/s)$ on this bounded range.  The coordinate-transport term is the Taylor increment recorded in Appendix Lemma~\ref{lem:g0-taylor-detail}.  Combining the two terms gives
\[
\begin{aligned}
        M_s(c)
        &=\eta\frac{(\log s)^2}{s}\frac{z}{2}\varphi_s'(z)
          +\eta\frac{\log s}{s}\left(\frac12\varphi_s''(z)-2\varphi_s(z)\right)
          +O_\lambda\left(\eta\frac{\log s}{s}\right).
\end{aligned}
\]
On $[\alpha_0,\alpha_3]$ one has $\varphi_s'(z)\ge0$, so the first term is nonnegative.  On $[\alpha_3,\alpha_5]$, the derivative bounds from Lemma~\ref{lem:g0-profile} and the exponential formula in \eqref{eq:g0-profile-def} give $|\varphi_s'(z)|\le C_\lambda/\log s$, so the first term is bounded below by $-C_\lambda\eta\log s/s$.  The remaining terms are also bounded below by $-C_\lambda\eta\log s/s$.  This proves the drift bound in \eqref{eq:g0-annulus-MD}; the bound on $D_s(c)$ again follows from \eqref{eq:g0-global-lipschitz-B}.

\emph{Regime 3.}  Suppose $z_s(c)\ge\alpha_5$.  Because $\alpha_5>\alpha_4$, all relevant predecessor arguments lie in the interval $z\ge\alpha_4$ for large $s$.  We give the calculation for $d<0$; the case $d>0$ is identical after reversing signs.  When $d<0$,
\[
        z_{s-1}(c)=\frac{|d|-q}{\sigma\sqrt{(s-1)\log(s-1)}},
        \qquad
        z_{s-1}(c-1)=\frac{|d|+1-q}{\sigma\sqrt{(s-1)\log(s-1)}}.
\]
Thus $z_{s-1}(c)<z_{s-1}(c-1)$.  Since $\varphi_s$ is decreasing on $z\ge\alpha_4$, $B_{s-1}(c)>B_{s-1}(c-1)$ and hence $D_s(c)>0$.  Since $d<0$, this proves $dD_s(c)<0$ in this case.

For the magnitude of the first difference, the identity
\[
        \varphi_{s-1}'(w)=-\frac{\lambda}{\log(s-1)}\varphi_{s-1}(w)
\]
and the mean-value theorem give
\[
\begin{aligned}
        D_s(c)
        &=\eta(\log(s-1))^2
          \{\varphi_{s-1}(z_{s-1}(c))-
             \varphi_{s-1}(z_{s-1}(c-1))\}  \\
        &=\eta(\log(s-1))^2
          \frac{\lambda}{\log(s-1)}
          \varphi_{s-1}(\widetilde z)
          \frac1{\sigma\sqrt{(s-1)\log(s-1)}}
\end{aligned}
\]
for some $\widetilde z$ between the two predecessor coordinates.  The two predecessor coordinates differ from $z$ by $O(1/(\sigma\sqrt{s\log s}))+O(z/s)$.  Thus, for the intermediate point $\widetilde z$ in the mean-value theorem,
\[
        |\widetilde z-z|
        \le C\left(\frac1{\sqrt{s\log s}}+\frac{z}{s}\right).
\]
In the range $z\le C_\rho\sqrt{s/(\log s)}$,
\[
        \frac{\lambda |\widetilde z-z|}{\log s}
        \le C_\lambda\left(\frac1{\sqrt{s}(\log s)^{3/2}}
        +\frac{z}{s\log s}\right)=o(1),
\]
uniformly.  Since the slope on $z\ge\alpha_4$ is $-\lambda/\log s$, this gives $\varphi_{s-1}(\widetilde z)/\varphi_s(z)=1+o(1)$ uniformly.  Therefore
\[
        |D_s(c)|
        =(1+o(1))\frac{\eta\lambda}{\sigma}
        \varphi_s(z)\sqrt{\frac{\log s}{s}},
\]
which is \eqref{eq:g0-tail-D}.

For the drift, let $Z=\frac{|d|+\xi}{\sigma\sqrt{(s-1)\log(s-1)}}$ with the sign-stable two-point increment from Lemma~\ref{lem:g0-sign-stable-detail}.  Then
\[
        \frac{M_s(c)}{\eta(\log s)^2\varphi_s(z)}
        =\frac{(\log(s-1))^2}{(\log s)^2}
          \E\left[\frac{\varphi_{s-1}(Z)}{\varphi_s(z)}\right]-1.
\]
Appendix Lemma~\ref{lem:g0-tail-ratio-detail} gives
\[
        \E\left[\frac{\varphi_{s-1}(Z)}{\varphi_s(z)}\right]
        =1-\frac{\lambda z}{2s\log s}
          +O_\lambda\left(\frac{1+z}{s(\log s)^2}+\frac1{s(\log s)^3}\right),
\]
while Lemma~\ref{lem:g0-s-asymp-detail} gives
\[
        \frac{(\log(s-1))^2}{(\log s)^2}
        =1-\frac{2}{s\log s}+O\left(\frac1{s^2\log s}\right).
\]
Multiplying the two displays yields
\[
        \frac{(\log(s-1))^2}{(\log s)^2}
          \E\left[\frac{\varphi_{s-1}(Z)}{\varphi_s(z)}\right]-1
        =-\frac{2+\lambda z/2+r_s(z)}{s\log s},
\]
where
\[
        \sup_{\alpha_5\le z\le C_\rho\sqrt{s/(\log s)}}
        \frac{|r_s(z)|}{1+\lambda z}\to0.
\]
Multiplying by $\eta(\log s)^2\varphi_s(z)$ proves \eqref{eq:g0-tail-M}.

\emph{Regime 4.}  If $|c/s-q|\ge\rho/2$ but some neighboring certificate value is nonzero, then by the support of $\Omega$ all relevant ratios are within $2\rho+O(1/s)$ of $q$.  In particular
\[
        z_s(c)\ge c_\rho\sqrt{s/(\log s)}.
\]
The exponential formula in \eqref{eq:g0-profile-def} then gives
\[
        \varphi_s(z_s(c))\le C\exp\{-c\sqrt{s}/(\log s)^{3/2}\},
\]
and the same bound holds at neighboring states.  Since $B$ and its first differences are polynomial factors times this super-polynomially small quantity, \eqref{eq:g0-macro-MD} follows.
\end{proof}

\subsection{Source and perturbation estimates}

The next estimate controls the Jensen/order-statistic term
\[
        \E\Delta(H_{s,c})-\Delta(\tau_s(c)).
\]
It is the bounded-density specialization of the estimate in Appendix~\ref{app:edge-order}.  We record the form used below. The later estimate in this subsection controls the combined term
$S_s(c,D)$, where the finite-difference perturbation is included.

\begin{lemma}[Quadratic bounded-density source on the moderate band]\label{lem:g0-inner-source}
There are constants $c_I>0$ and $s_I<\infty$ such that, for all $s\ge s_I$ and all $c$ with $|x_s(c)|\le \alpha_5\sqrt{\log s}+1$,
\begin{equation}\label{eq:g0-inner-source}
        \E\Delta(H_{s,c})
        \ge c_I\frac{(1+|x_s(c)|)^2}{s}.
\end{equation}
In particular, if $z_s(c)\ge\alpha_0$, then
\begin{equation}\label{eq:g0-inner-source-ell}
        \E\Delta(H_{s,c})
        \ge c_I \alpha_0^2 \frac{\log s}{s}.
\end{equation}
\end{lemma}

\begin{proof}
Apply Lemma~\ref{lem:gapped-random-rank-source} with $A=\alpha_5<1$ and $L=1$.  The two-uniform model has bounded positive densities at the two support edges adjacent to the gap, and the appendix estimate reduces exactly to \eqref{eq:g0-inner-source}.  If $z_s(c)\ge\alpha_0$, then $|x_s(c)|=z_s(c)\sqrt{\log s}\ge\alpha_0\sqrt{\log s}$, proving \eqref{eq:g0-inner-source-ell}.
\end{proof}

The range $z_s(c)\ge\alpha_5$ requires a bound on the combined term $S_s(c,D)$ that uses the sign condition $d_s(c)D\le0$.  This is the only such estimate in this section not already contained in the appendix estimates for the vanishing-density case.  The following elementary comparison explains the constants used in the large- and small-overshoot alternatives.

\begin{lemma}[Deterministic bounded-density comparisons]\label{lem:g0-deterministic-comparison}
Let $R,A>0$, $0\le D\le RA/8$, $x\ge0$, and $m_2\ge0.99R^2A^2$.  Then there is an absolute constant $c_{\rm det}>0$ such that:
\begin{enumerate}[label=(\roman*),leftmargin=2.5em]
\item if $x\le3RA/4$, then $
        m_2-((x-D)_+)^2\ge c_{\rm det}R^2A^2;
$
\item if $x>3RA/4$, then
$        x^2-((x-D)_+)^2\ge \frac{11}{8}RAD.
$
\end{enumerate}
\end{lemma}

\begin{proof}
If $x\le3RA/4$, then $(x-D)_+\le x+D\le7RA/8$.  Hence
\[
        m_2-((x-D)_+)^2
        \ge \left(0.99-\left(\frac78\right)^2\right)R^2A^2,
\]
and the parenthetical constant is positive.  If $x>3RA/4$, then $u_+^2\le u^2$ gives
\[
        x^2-((x-D)_+)^2\ge x^2-(x-D)^2=2xD-D^2.
\]
Using $x>3RA/4$ and $D\le RA/8$ yields
\[
        2xD-D^2\ge \frac32 RAD-\frac18RAD=\frac{11}{8}RAD.
\]
\end{proof}

\begin{lemma}[Uniform bounded-density bound for $S_s(c,D)$]\label{lem:g0-tail-source}
There exist constants $K=11/16>1/2$, $c_0>0$, and $\varepsilon_0>0$, depending only on the two-uniform distribution and on $\rho$, with the following property. Fix $\alpha>0$. For all sufficiently large $s$, suppose
\[
        \alpha \sigma\sqrt{s\log s}\le |d_s(c)|\le 2\rho s,
        \qquad
        a_s(c):=\frac{|d_s(c)|}{s},
\]
and
\[
        |D|\le\varepsilon_0 a_s(c), \qquad d_s(c)D\le0.
\]
Define
\[
(w,R):=
\begin{cases}
\bigl((\tau_s(c)-a_2)_+,\,L_+/q\bigr), & d_s(c)<0,\\
\bigl((b_1-\tau_s(c))_+,\,L_-/p\bigr), & d_s(c)>0.
\end{cases}
\]
Then, if $w>\frac34 R a_s(c)$,
\[
        \E\Delta(H_{s,c})-\Delta(\tau_s(c)-D)
        \ge K a_s(c)|D|,
\]
while, if $w\le\frac34 R a_s(c)$,
\[
        \E\Delta(H_{s,c})-\Delta(\tau_s(c)-D)
        \ge c_0 a_s(c)^2.
\]
\end{lemma}

\begin{proof}
We prove the lemma when $d_s(c)<0$; the other side is symmetric.  Put $a:=a_s(c)$, $\tau=\tau_s(c)$, $\widehat\tau=\tau-D$, $X=(H_{s,c}-a_2)_+$, and $w=(\tau-a_2)_+$.  We choose $\varepsilon_0$ below fixed distributional constants; in particular,
\[
        \varepsilon_0\le \min\left\{\frac{G}{4\rho},\frac{R_+}{8},\frac{R_-}{8}\right\}.
\]
Since $D\ge0$ and $|D|\le\varepsilon_0 a\le 2\rho\varepsilon_0$, this gives $D<G/2$ throughout the stated range.  We first check that $\widehat\tau\in[b_1,b_2]$ for all large $s$.

The upper bound $\widehat\tau\le b_2$ is immediate from $\tau\le b_2$ and $D\ge0$.  For the lower bound, if $\tau\ge a_2$, then $\widehat\tau\ge a_2-D>b_1+G/2$.  Suppose instead that $\tau<a_2$.  Since $H_{s,c}\ge a_2$ on $\{N\ge c\}$ and $H_{s,c}\ge a_1$ always,
\[
        \tau=\E H_{s,c}
        \ge a_2\Prob(N\ge c)+a_1\Prob(N<c)
        =a_2-(a_2-a_1)\Prob(N<c).
\]
Because $d_s(c)<0$ and $a=|d_s(c)|/s$, $\mu_s-c=a\,s-q$.  For large $s$, $a\,s\to\infty$, and hence
\[
        a\,s-q+1\ge \frac12 a\,s.
\]
Bernstein's inequality for the binomial variable $N$ gives
\[
        \Prob(N<c)
        \le \Prob\{\mu_s-N\ge a\,s-q+1\}
        \le \exp\left\{-\frac{c'_B a^2s^2}{s+a\,s}\right\}.
\]
Since $a\le2\rho$ in the stated range, the denominator is at most a constant multiple of $s$, and therefore
\[
        \Prob(N<c)\le \exp\{-c_B a^2s\}
        \le s^{-c_B\alpha^2},
        \qquad a\in[\alpha\sigma\sqrt{(\log s)/s},2\rho],
\]
where $c_B>0$ depends only on $q$ and $\rho$.  The constant $\alpha$ is fixed before $s\to\infty$, so the probability is $o(1)$ uniformly over this range; at the lower end of the range it is polynomially small in $s$, which is sufficient here.  Hence $\tau\ge a_2-G/4=b_1+3G/4$ for all large $s$.  Since $D<G/2$, this gives $\widehat\tau\ge b_1+G/4$.  Thus $\widehat\tau\in[b_1,b_2]$.

Using the exact convex-residual formula \eqref{eq:g0-Delta-exact}, and because $\widehat\tau\in[b_1,b_2]$,
\begin{equation}
        \Delta(\widehat\tau)
        =\frac{q}{2L_+}\bigl((w-D)_+\bigr)^2,
\end{equation}
while
\begin{equation}
        \E\Delta(H_{s,c})
        \ge\frac{q}{2L_+}\E X^2.
\end{equation}
Therefore
\begin{equation}\label{eq:g0-tail-master-bound}
        \E\Delta(H_{s,c})-\Delta(\widehat\tau)
        \ge\frac{q}{2L_+}
        \left(\E X^2-((w-D)_+)^2\right).
\end{equation}

We need a lower bound on $\E X^2$.  Conditional on $N=n\ge c$, $X$ is the order statistic of rank $n-c+1$ among $n$ iid uniform variables on $[0,L_+]$.  Hence
\[
        \E[X^2\mid N=n]
        =L_+^2\frac{(n-c+1)(n-c+2)}{(n+1)(n+2)}.
\]
Because $d_s(c)<0$ and $a=|d_s(c)|/s$, one has $\mu_s-c=a\,s-q$.  Fix $\delta\in(0,q/4)$.  Let
\[
        \mathcal G_s:=\bigl\{N-c\ge (1-\delta)a\,s\bigr\}
        \cap
        \bigl\{|N-\mu_s|\le \delta s\bigr\}.
\]
The first event in $\mathcal G_s$ fails only if
\[
        \mu_s-N>\delta a\,s-q.
\]
Since $a\,s\to\infty$ uniformly over $a\in[\alpha\sigma\sqrt{(\log s)/s},2\rho]$, the right side is at least $(\delta/2)a\,s$ for all large $s$.  Bernstein's inequality then gives
\[
        \Prob\bigl(N-c<(1-\delta)a\,s\bigr)
        \le \exp\{-c_{B,\delta}a^2s\}=o(1)
\]
uniformly over the same range of $a$.  The second event fails with probability at most $\exp\{-c_\delta s\}$.  Hence $\Prob(\mathcal G_s)=1-o(1)$ uniformly.  On $\mathcal G_s$, for all large $s$,
\[
        N-c+1\ge (1-\tfrac32\delta)a\,s,
        \qquad
        N-c+2\ge (1-\tfrac32\delta)a\,s,
        \qquad
        N+1,N+2\le (q+2\delta)s.
\]
Substituting these deterministic bounds into the conditional second-moment formula gives, on $\mathcal G_s$,
\[
        \E[X^2\mid N]
        \ge L_+^2
        \frac{(1-\tfrac32\delta)^2a^2s^2}
             {(q+2\delta)^2s^2}.
\]
Using the tower property and $X^2\ge0$ on the complement of $\mathcal G_s$ yields
\[
\begin{aligned}
        \E X^2
        &=\E\bigl[\E[X^2\mid N]\bigr]\\
        &\ge \E\bigl[\E[X^2\mid N]\1_{\mathcal G_s}\bigr]\\
        &\ge
        L_+^2
        \frac{(1-\tfrac32\delta)^2a^2s^2}
             {(q+2\delta)^2s^2}
        \Prob(\mathcal G_s)\\
        &\ge L_+^2\frac{(1-\tfrac32\delta)^2}{(q+2\delta)^2}a^2(1-o(1)).
\end{aligned}
\]
This estimate is uniform on
$a\in[\alpha\sigma\sqrt{(\log s)/s},2\rho]$: at the lower end of the range the deviations are only logarithmic, but $a^2s\asymp\log s$ still forces $\Prob(\mathcal G_s)\to1$.  Since the prefactor tends to $(L_+/q)^2=R_+^2$ as $\delta\downarrow0$, we may choose $\delta$ and then $s$ large enough so that, for any prescribed $\varepsilon>0$,
\begin{equation}
        \E X^2\ge (1-\varepsilon)R_+^2a^2.
\end{equation}
Take $\varepsilon=0.01$.

By the choice of $\varepsilon_0$, $D\le R_+a/8$.  If $w\le\frac34R_+a$, apply Lemma~\ref{lem:g0-deterministic-comparison} with $R=R_+$ and $m_2=\E X^2$.  Together with \eqref{eq:g0-tail-master-bound}, this gives
\[
        \E\Delta(H_{s,c})-\Delta(\widehat\tau)
        \ge c_0a^2
\]
for a positive constant $c_0$.

If $w>\frac34R_+a$, then $w=\tau-a_2>0$.  The pointwise inequality $X=(H_{s,c}-a_2)_+\ge H_{s,c}-a_2$ holds also on the off-event $\{N<c\}$, where $X=0$ and $H_{s,c}-a_2<0$.  Therefore $\E X\ge \E(H_{s,c}-a_2)=w$, and Jensen's inequality gives $\E X^2\ge w^2$.  Lemma~\ref{lem:g0-deterministic-comparison}, again with $R=R_+$, gives
\[
        w^2-((w-D)_+)^2\ge\frac{11}{8}R_+aD.
\]
Substituting into \eqref{eq:g0-tail-master-bound} and using $qR_+/(2L_+)=1/2$ yields
\[
        \E\Delta(H_{s,c})-\Delta(\widehat\tau)
        \ge \frac{11}{16}aD.
\]
This proves the case $d_s(c)<0$ with $K=11/16$.

For the case $d_s(c)>0$, set $Y=(b_1-H_{s,c})_+$ and $w=(b_1-\tau)_+$.  Now $D\le0$, so $|D|=-D$.  The same gap-margin argument shows $\widehat\tau\in[a_1,a_2]$ for all large $s$.  On this side the exact convex-residual formula gives
\[
        \Delta(\widehat\tau)=\frac{p}{2L_-}\bigl((w - |D|)_+\bigr)^2,
        \qquad
        \E\Delta(H_{s,c})\ge \frac{p}{2L_-}\E Y^2.
\]
Conditional on $N=n<c$, the distance $Y$ has the same uniform-order-statistic second-moment formula as above, with $L_-$ and $p$ replacing $L_+$ and $q$ and rank $c-n$.  Repeating the preceding tower-property argument gives $\E Y^2\ge(1-\varepsilon)R_-^2a^2$, uniformly on the stated range.  Applying Lemma~\ref{lem:g0-deterministic-comparison} with $R=R_-$ and $D$ replaced by $|D|$ proves the stated linear and quadratic alternatives for the case $d_s(c)>0$.
\end{proof}

For states outside the central range, the proof uses the prophet Jensen slack rather than the signed first-difference source.

\begin{lemma}[Jensen slack in the outer cutoff range]\label{lem:g0-macro-slack}
There are constants $\kappa_\rho>0$ and $s_\rho<\infty$ such that, for all $s\ge s_\rho$ and all interior $c$ satisfying
\begin{equation}\label{eq:g0-active-macro-band}
        \rho/2\le |c/s-q|\le2\rho+1/s,
\end{equation}
one has
\begin{equation}
        \delta_s(c)
        :=\E\Delta(H_{s,c})-\Delta(\tau_s(c))
        \ge \frac{\kappa_\rho}{s}.
\end{equation}
\end{lemma}

\begin{proof}
We prove the high-deficit case $\rho/2\le q-c/s\le2\rho+1/s$; the other side is symmetric.  Let $E_s:=\{|N-\mu_s|\le(\rho/4)s\}$.  By Bernstein's inequality, $\Prob(E_s^c)\le e^{-c_\rho s}$.  On $E_s$, for large $s$,
\[
        N-c\ge \frac{\rho}{8}s,
        \qquad
        N\le(q+\rho/4)s,
        \qquad
        c\ge(q-2\rho)s-1\ge\frac q2 s.
\]
Thus $H_{s,c}$ lies in the upper support on $E_s$, and conditional on $N$ it is a uniform order statistic with rank $r_N=N-c+1$.  The variance formula for uniform order statistics gives
\[
        \Var(H_{s,c}\mid N)
        =L_+^2\frac{r_N(N-r_N+1)}{(N+1)^2(N+2)}
        \ge \frac{\kappa_1}{s}
        \qquad\text{on }E_s.
\]
On the upper interval, $\Delta(\tau)=q(\tau-a_2)^2/(2L_+)$, so the conditional Jensen gap equals
\[
        \E[\Delta(H_{s,c})\mid N]
        -\Delta(\E[H_{s,c}\mid N])
        =\frac{q}{2L_+}\Var(H_{s,c}\mid N)
        \ge \frac{\kappa_2}{s}
\]
on $E_s$.  Decomposing the unconditional Jensen gap into the expectation of conditional Jensen gaps plus the Jensen gap of the conditional mean, the second term is nonnegative by convexity.  Hence
\[
        \delta_s(c)\ge \frac{\kappa_2}{s}\Prob(E_s)
        \ge\frac{\kappa_2}{2s}
\]
for large $s$.
\end{proof}

\subsection{Certificate verificatioin}

We now verify the Bellman certificate inequality \eqref{eq:common-bellman-residual} for the certificate \eqref{eq:g0-certificate}.  Recall that
\[
        S_s(c,D)=\E\Delta(H_{s,c})-
        \Delta(\tau_s(c)-D).
\]
The proof separates four regimes.  In the central range $z_s(c)\le\alpha_0$, the drift is nonnegative and the perturbation remains inside the gap.  For $\alpha_0\le z_s(c)\le\alpha_5$, the Jensen/order-statistic source dominates the bounded negative drift.  For $z_s(c)\ge\alpha_5$ near $q$, Lemma~\ref{lem:g0-tail-source} controls the combined term $S_s(c,D_s(c))$.  Outside the central range, the certificate terms are $o(1/s)$ and are dominated by the prophet Jensen slack.

\begin{proposition}[Feasibility of the bounded-density certificate]\label{prop:g0-feasible}
There exist $\eta_0>0$ and $s_0<\infty$ such that, for every $\eta\in(0,\eta_0]$, the certificate \eqref{eq:g0-certificate} satisfies
\begin{equation}\label{eq:g0-bellman-residual}
        M_s(c)+S_s(c,D_s(c))\ge0
\end{equation}
for every $s>s_0$ and every interior state $1\le c<s$.
\end{proposition}

\begin{proof}
We choose $\eta_0$ small enough for the applications of Lemma~\ref{lem:g0-tail-source} and for the comparisons below, and then choose $s_0$ large enough so that all asymptotic estimates hold uniformly for $\eta\le\eta_0$.  Fix such $s$ and $c$.

\emph{Regime 1: $z_s(c)\le\alpha_0$ and $|c/s-q|\le\rho/2$.}  Lemma~\ref{lem:g0-fdiff} gives $M_s(c)\ge0$ and $|D_s(c)|\le C\eta(\log s)^{3/2}/\sqrt{s}$.  The bounded-density case of the common gap-margin Lemma~\ref{lem:gapped-gap-margin-pub}, used with any $A\in(\alpha_5,1)$, gives
\[
        \dist(\tau_s(c),\{b_1,a_2\})
        \ge c s^{-A^2/2}(\log s)^{-1/2}
\]
throughout $z_s(c)\le\alpha_5$.  Since $A<1$,
\[
        \frac{(\log s)^{3/2}/\sqrt{s}}{s^{-A^2/2}(\log s)^{-1/2}}
        =s^{-(1-A^2)/2}(\log s)^2\to0.
\]
Thus $\tau_s(c)-D_s(c)$ lies in the gap for all large $s$, so $\Delta(\tau_s(c)-D_s(c))=0$ and $S_s(c,D_s(c))=\E\Delta(H_{s,c})\ge0$.  Hence \eqref{eq:g0-bellman-residual} holds.

\emph{Regime 2: $\alpha_0\le z_s(c)\le\alpha_5$ and $|c/s-q|\le\rho/2$.}  As in Regime 1, the perturbation $D_s(c)$ is smaller than the gap margin, hence $S_s(c,D_s(c))=\E\Delta(H_{s,c})$.  Lemma~\ref{lem:g0-inner-source} gives $S_s(c,D_s(c))\ge c_1\log s/s$, while Lemma~\ref{lem:g0-fdiff} gives $M_s(c)\ge -C_\lambda\eta\log s/s$.  Choosing $\eta_0\le c/(2C_\lambda)$ yields \eqref{eq:g0-bellman-residual} in this regime.

\emph{Regime 3: $z_s(c)\ge\alpha_5$ and $|c/s-q|\le\rho/2$.}  Put $z=z_s(c)$ and $a_s(c)=|d_s(c)|/s=\sigma z\sqrt{\log s/s}$.  By Lemma~\ref{lem:g0-fdiff}, $d_s(c)D_s(c)<0$ and
\[
        a_s(c)|D_s(c)|
        =(1+o(1))\eta\lambda z\varphi_s(z)\frac{\log s}{s}.
\]
We apply Lemma~\ref{lem:g0-tail-source} with the fixed lower-tail parameter $\alpha=\alpha_5$.
Also,
\[
        \frac{|D_s(c)|}{a_s(c)}
        =(1+o(1))\frac{\eta\lambda \varphi_s(z)}{\sigma^2 z}
        \le C\eta\lambda,
\]
because $\varphi_s$ is bounded and $z\ge\alpha_5$.  Reducing $\eta_0$ ensures $|D_s(c)|\le\varepsilon_0a_s(c)$, so Lemma~\ref{lem:g0-tail-source} applies.

In Lemma~\ref{lem:g0-tail-source}, if $w>\frac34 R a_s(c)$, then $S_s(c,D_s(c))\ge K a_s(c)|D_s(c)|$ with $K=11/16$.  Combining with \eqref{eq:g0-tail-M},
\[
\begin{aligned}
        M_s(c)+S_s(c,D_s(c))
        &\ge \eta \varphi_s(z)\frac{\log s}{s}
        \left[(K-1/2)\lambda z-2-r_s(z)\right].
\end{aligned}
\]
By \eqref{eq:g0-lambda-choice}, the deterministic part of the bracket satisfies
\[
        (K-1/2)\lambda z-2
        \ge (K-1/2)\lambda\alpha_5-2>1
        \qquad (z\ge\alpha_5).
\]
The uniform bound on $r_s(z)$ in Lemma~\ref{lem:g0-fdiff} implies that $r_s(z)=o(1)$ on bounded subranges and $r_s(z)=o(\lambda z)$ on growing subranges, uniformly over $\alpha_5\le z\le C_\rho\sqrt{s/(\log s)}$.  Therefore the lower bound for $S_s(c,D_s(c))$ dominates the error for all sufficiently large $s$, and the bracket is positive uniformly.

If $w\le\frac34 R a_s(c)$, Lemma~\ref{lem:g0-tail-source} gives
\[
        S_s(c,D_s(c))\ge c_0a_s(c)^2
        =c_0\sigma^2 z^2\frac{\log s}{s}.
\]
On $z\ge\alpha_5$, $\varphi_s(z)$ is bounded and decreasing.  Hence, for $z\in[\alpha_5,1]$, the ratio $\varphi_s(z)(1+\lambda z)/z^2$ is bounded by a constant depending only on $\lambda,\alpha_5$ and $\varphi_\star$.  For $z\ge1$, $1+\lambda z\le(1+\lambda)z^2$, and $\varphi_s(z)\le \varphi_s(\alpha_5)\le C \varphi_\star$.  Therefore, for all $z\ge\alpha_5$,
\[
        \varphi_s(z)(1+\lambda z)\le C_{\lambda,\alpha_5}z^2.
\]
Using \eqref{eq:g0-tail-M},
\[
        -M_s(c)\le C\eta \varphi_s(z)(1+\lambda z)\frac{\log s}{s}
        \le C'\eta z^2\frac{\log s}{s}.
\]
Reducing $\eta_0$ once more gives \eqref{eq:g0-bellman-residual} in the small-overshoot subcase.

\emph{Regime 4: $|c/s-q|\ge\rho/2$.}  If all three values $B_s(c)$, $B_{s-1}(c)$, and $B_{s-1}(c-1)$ vanish, then $M_s(c)=D_s(c)=0$, and
\[
        S_s(c,0)=\E\Delta(H_{s,c})-\Delta(\tau_s(c))=\delta_s(c)\ge0
\]
by the Jensen-slack identity.

If at least one neighboring certificate value is nonzero, then the support of $\Omega$ implies $|c/s-q|\le2\rho+1/s$ for large $s$, so \eqref{eq:g0-active-macro-band} holds.  Lemma~\ref{lem:g0-fdiff} gives $|M_s(c)|+|D_s(c)|=o(1/s)$, while Lemma~\ref{lem:g0-macro-slack} gives $\delta_s(c)\ge\kappa_\rho/s$.  Since $h$ is globally Lipschitz for a compactly supported distribution and $h_0$ is affine, $\Delta=h-h_0$ is globally Lipschitz. Therefore the following comparison is valid even if the $o(1/s)$ perturbation moves $\tau_s(c)-D_s(c)$ slightly outside $[a_1,b_2]$:
\[
        S_s(c,D_s(c))
        =\delta_s(c)+\Delta(\tau_s(c))-
        \Delta(\tau_s(c)-D_s(c))
        \ge \frac{\kappa_\rho}{s}-C_\Delta |D_s(c)|
        = \frac{\kappa_\rho}{s}-o(1/s).
\]
Together with $M_s(c)=o(1/s)$, this proves \eqref{eq:g0-bellman-residual}.  The four regimes exhaust all interior states, so the proposition follows.
\end{proof}

\subsection{The lower bound}

\begin{theorem}[Two-uniform mixture]\label{thm:gapped-zero}
Assume \eqref{eq:g0-uniform-model}.  Let $k_T=\lfloor qT\rfloor$.  There is a constant $c>0$, depending only on the two-uniform distribution, such that for all sufficiently large $T$,
\begin{equation}
        B_T^\star(k_T)
        =\Phi_T(k_T)-J_T(k_T)
        \ge c(\log T)^2.
\end{equation}
\end{theorem}
\begin{proof}
Let $\eta_0$ and $s_0$ be as in Proposition~\ref{prop:g0-feasible}.  By Corollary~\ref{cor:base-scale-framework}, after possibly reducing $\eta\in(0,\eta_0]$, the certificate satisfies the base condition at time $s_0$.  Proposition~\ref{prop:g0-feasible} verifies the Bellman certificate inequality for all later times, so $B$ is feasible for $\mathsf P(T,k_T,s_0)$ for all sufficiently large $T$.

By Proposition~\ref{prop:common-comparison},
\[
        \Reg(T,k_T;F)\ge B_T(k_T).
\]
Now $|k_T-qT|\le1$, hence $z_T(k_T)=O((T\log T)^{-1/2})$.  For large $T$, $\Omega(k_T/T)=1$ and $z_T(k_T)\le\alpha_1$, so \eqref{eq:g0-profile-def} gives
\[
        B_T(k_T)
        =\eta(\log T)^2\left(f_0+\kappa z_T(k_T)^2\right)
        =\eta f_0(\log T)^2+o(1).
\]
Therefore $\Reg(T,k_T;F)\ge(\eta f_0/2)(\log T)^2$ for all sufficiently large $T$.
\end{proof}

\section{Conclusion}\label{sec:conclusion}

This paper proves a tight lower bound for the additive prophet-online regret in the bounded-density gapped case of the multi-secretary problem: for a mixture of two separated uniform distributions, the regret at the critical capacity is $\Omega((\log T)^2)$.  This shows that the additional logarithmic factor in prior upper bounds for gapped bounded-density instances is not an artifact of the analysis, but is already unavoidable in the one-resource multi-secretary specialization.  The same Bellman-certificate framework also yields a matching lower bound for gapped distributions with gap-facing vanishing densities in the shifted-capacity regime; this companion result is proved in Appendix~\ref{app:gapped-positive}.

The proof approach is based on explicit feasible Bellman certificates for the exact Bellman recursion of the prophet-online regret.  The certificates make visible the mechanism behind the regret order of the gapped support case.  With a support gap, the reward curve is locally affine at the operating threshold and curvature appears only near the adjacent support edges; the certificate must therefore exploit moderate-deviation fluctuations of the prophet's marginal order statistic.  This separation between where slack is generated and where the certificate is transported is what permits the log-squared lower bound in the bounded-density gapped case.  The Bellman-certificate method therefore works directly with the regret recursion rather than through policy-level indistinguishability arguments.

\appendix
\section{Binomial estimates}\label{app:binomial}\label{sec:binomial}

This section collects the binomial estimates used throughout the proof, stated uniformly over the central and moderate-deviation ranges that arise in the certificates.  The estimates are elementary consequences of Bernstein's inequality, Berry--Esseen, Mills' ratio, and moment bounds for sums of bounded independent random variables; see, for example, \citet{BoucheronLugosiMassart2013} for Bernstein-type concentration inequalities and \citet{Feller1971} for classical normal approximation and tail estimates.

Throughout this section fix $q\in(0,1)$ and write $\sigma^2:=q(1-q)$.  For each horizon $s\ge2$, let
\[
        N=N_s\sim \mathrm{Bin}(s-1,q),
        \qquad
        m:=s-1,
        \qquad
        \mu_s:=q(s-1).
\]
For an integer capacity $c\in\{0,\ldots,s\}$, use the coordinate $x_s(c)$ from \eqref{eq:coords-model} and set $u_s(c):=|x_s(c)|$.  When no confusion is possible we write $x=x_s(c)$ and $u=|x|$.  Notice that the binomial mean is $\mu_s=qs-q$, whereas the state coordinate is centered at $qs$.  This distinction creates only $O(s^{-1/2})$ shifts.  More precisely, for any integer $c$,
\begin{align}
        \frac{c-\mu_s}{\sigma\sqrt{m}}
        &=x_s(c)\sqrt{\frac{s}{m}}+\frac{q}{\sigma\sqrt{m}},\label{eq:binom-center-shift-1}\\
        \frac{c-1-\mu_s}{\sigma\sqrt{m}}
        &=x_s(c)\sqrt{\frac{s}{m}}-\frac{1-q}{\sigma\sqrt{m}}.\label{eq:binom-center-shift-2}
\end{align}
Thus, uniformly in the moderate-deviation range $|x_s(c)|\le A\sqrt{\log s}+L$,
\begin{equation}
        \frac{c-\mu_s}{\sigma\sqrt{m}}=x_s(c)+O\!\left(\frac{1+|x_s(c)|}{s}+\frac1{\sqrt{s}}\right),
\end{equation}
and the same estimate holds with $c$ replaced by $c-1$.

\subsection{Concentration and central moments}

We begin with a nonasymptotic concentration estimate.  It is useful both for tail truncations and for proving uniform moment bounds.

\begin{lemma}[Bernstein bound for the binomial]\label{lem:bernstein-binomial}
For every $s\ge2$ and every $t\ge0$,
\begin{align}
        \Prob(N-\mu_s\ge t)
        &\le
        \exp\!\left(-\frac{t^2}{2\sigma^2(s-1)+2t/3}\right),\label{eq:bernstein-upper}\\
        \Prob(\mu_s-N\ge t)
        &\le
        \exp\!\left(-\frac{t^2}{2\sigma^2(s-1)+2t/3}\right).\label{eq:bernstein-lower}
\end{align}
Consequently, for every $b>0$ there is $C_b<\infty$ such that
\begin{equation}\label{eq:binom-central-moment}
        \E|N-\mu_s|^b\le C_b s^{b/2}
        \qquad\text{for all }s\ge2.
\end{equation}
\end{lemma}

\begin{proof}
Write $N-\mu_s=\sum_{i=1}^{s-1}(Z_i-q)$, where $Z_i$ are iid Bernoulli$(q)$.  The summands are mean zero, bounded by $1$ in absolute value, and have total variance $(s-1)\sigma^2$.  Bernstein's inequality gives \eqref{eq:bernstein-upper}; applying the same inequality to $-(N-\mu_s)$ gives \eqref{eq:bernstein-lower}.

For the moment bound, use the tail-integral identity
\[
        \E|N-\mu_s|^b
        = b\int_0^\infty t^{b-1}\Prob(|N-\mu_s|\ge t)\,dt.
\]
Split the integral at $t=s$.  On $0\le t\le s$, Bernstein gives a sub-Gaussian bound after increasing constants:
\[
        \Prob(|N-\mu_s|\ge t)
        \le 2\exp\!\left(-c\frac{t^2}{s}\right)
\]
for a constant $c=c(q)>0$.  Integrating this bound gives $O(s^{b/2})$.  On $t>s$, the probability is zero because $|N-\mu_s|\le s-1$.  This proves \eqref{eq:binom-central-moment}.
\end{proof}

A useful corollary is the following square-root logarithmic concentration estimate.

\begin{corollary}[Square-root logarithmic concentration]\label{cor:binom-sqrtlog}
For every $a>0$ there are constants $C_a<\infty$ and $s_a<\infty$ such that, for all $s\ge s_a$,
\begin{equation}\label{eq:binom-sqrtlog}
        \Prob\left(|N-\mu_s|\ge C_a\sqrt{s\log s}\right)\le s^{-a}.
\end{equation}
\end{corollary}

\begin{proof}
Apply Lemma~\ref{lem:bernstein-binomial} with $t=C\sqrt{s\log s}$.  For fixed $C$, the denominator in Bernstein's exponent is $2\sigma^2s+o(s)$, hence the exponent is $-(C^2/(2\sigma^2)+o(1))\log s$.  Choosing $C$ sufficiently large gives \eqref{eq:binom-sqrtlog}.
\end{proof}

\subsection{Moderate-deviation tail lower bounds}

The next lemma gives the lower-tail estimates needed to keep the offline marginal $H_{s,c}$ inside the gap.  

\begin{lemma}[Moderate-deviation tail lower bound]\label{lem:binom-tail-common}
Fix constants $A\in(0,1)$ and $L<\infty$.  There are constants $a_A>0$ and $s_A<\infty$, depending only on $q,A,L$, such that for all $s\ge s_A$ and all integers $c$ with $u=|x_s(c)|\le A\sqrt{\log s}+L$, 
\begin{equation}\label{eq:binom-min-tail}
        \min\{\Prob(N<c),\Prob(N\ge c)\}
        \ge
        a_A\frac{e^{-u^2/2}}{1+u}.
\end{equation}
\end{lemma}

\begin{proof}
Let $\Phi^N$ denote the standard normal distribution function.  The Berry--Esseen theorem for Bernoulli sums gives a constant $C_{\rm BE}=C_{\rm BE}(q)$ such that, for all real $t$,
\begin{equation}\label{eq:BE-binomial}
        \left|
        \Prob\left(\frac{N-\mu_s}{\sigma\sqrt{m}}\le t\right)
        -\Phi^N(t)
        \right|
        \le \frac{C_{\rm BE}}{\sqrt{m}}.
\end{equation}
We first lower-bound $\Prob(N<c)$ when $x_s(c)\le0$.  Let $u=|x_s(c)|$.  By \eqref{eq:binom-center-shift-2},
\[
        \frac{c-1-\mu_s}{\sigma\sqrt{m}}
        =-u+\eps_s,
        \qquad
        |\eps_s|\le C\left(\frac{1+u}{s}+\frac1{\sqrt{s}}\right).
\]
In the range $u\le A\sqrt{\log s}+L$, one has $|\eps_s|=o(1)$ and $(1+u)|\eps_s|=o(1)$.  The mean-value theorem applied to $v\mapsto\log\Phi^N(-v)$ gives
\[
        \left|\log\frac{\Phi^N(-u+\eps_s)}{\Phi^N(-u)}\right|
        \le C(1+u)|\eps_s|=o(1),
\]
where the derivative bound uses Mills' ratio.  Hence the small shift changes the normal tail by a factor $1+o(1)$ uniformly over this range.  Applying Mills' ratio once more gives
\begin{equation}\label{eq:normal-tail-shift-left}
        \Phi^N(-u+\eps_s)
        \ge
        c_1\Phi^N(-u)
        \ge
        c_2\frac{e^{-u^2/2}}{1+u}
\end{equation}
for constants $c_1,c_2>0$.  Since $A<1$,
\[
        \frac{e^{-u^2/2}}{1+u}
        \ge
        c_3\frac{s^{-A^2/2}\exp(-A L\sqrt{\log s})}{1+\sqrt{\log s}},
\]
and this quantity dominates $s^{-1/2}$.  Therefore the Berry--Esseen error in \eqref{eq:BE-binomial} is at most one half of the right-hand side of \eqref{eq:normal-tail-shift-left} for all sufficiently large $s$.  Since $\Prob(N<c)=\Prob(N\le c-1)$, this gives the required lower bound on $\Prob(N<c)$ when $x_s(c)\le0$.

The corresponding lower bound on $\Prob(N\ge c)$ when $x_s(c)\ge0$ is identical.  In that case,
\[
        \frac{c-\mu_s}{\sigma\sqrt{m}}=u+o(1),
\]
and applying Berry--Esseen to the upper tail, with the same Mills-ratio comparison, gives
\[
        \Prob(N\ge c)
        \ge c_4\frac{e^{-u^2/2}}{1+u}
\]
for a constant $c_4>0$.

Finally, \eqref{eq:binom-min-tail} follows by combining the preceding estimates with the trivial observation that the non-rare side has probability bounded below by a positive constant.  For example, if $x_s(c)\le0$, then the preceding lower bound controls $\Prob(N<c)$, while $\Prob(N\ge c)$ is at least a positive constant for all large $s$ by Berry--Esseen and \eqref{eq:binom-center-shift-1}.  The case $x_s(c)\ge0$ is symmetric.  Reducing $a_A$ if necessary gives \eqref{eq:binom-min-tail}.
\end{proof}

The form most often used later is the following direct consequence.

\begin{corollary}[Tail lower bound on the $\sqrt{s\log s}$ scale]\label{cor:binom-tail-sqrtlog}
Fix $A\in(0,1)$.  There are constants $b_A>0$ and $s_A<\infty$ such that, whenever $s\ge s_A$ and
$|c-qs|\le A\sigma\sqrt{s\log s}$,
one has
\begin{equation}
        \min\{\Prob(N<c),\Prob(N\ge c)\}
        \ge
        b_A s^{-A^2/2}(\log s)^{-1/2}.
\end{equation}
\end{corollary}

\begin{proof}
Apply Lemma~\ref{lem:binom-tail-common} with $L=0$.  Since $u\le A\sqrt{\log s}$,
\[
        \frac{e^{-u^2/2}}{1+u}
        \ge
        \frac{s^{-A^2/2}}{1+A\sqrt{\log s}}
        \ge
        b_A s^{-A^2/2}(\log s)^{-1/2}
\]
for all large $s$.
\end{proof}

When a bound has the form $|x_s(c)|\le A_0\sqrt{\log s}+O(1)$, we will apply Corollary~\ref{cor:binom-tail-sqrtlog} with any fixed $A>A_0$; this absorbs the $O(1)$ buffer for all sufficiently large $s$.

\subsection{Positive-part binomial moments}

The source estimates in the vanishing-density sections require moments of the random rank by which the future high count exceeds, or falls short of, the capacity.  The next lemma gives both upper bounds and truncated lower bounds.  The truncation in the lower bound will later ensure that the relevant order-statistic ranks are at most a constant multiple of $(1+u)\sqrt{s}$.

\begin{lemma}[Positive-part binomial moments]\label{lem:binomial-moment-common}
Fix $a>0$ and constants $A,L<\infty$.  There are constants $0<b_a<B_a<\infty$ and $s_a<\infty$, depending only on $q,a,A,L$, such that for all $s\ge s_a$ and all integers $c$ with
$u= |x_s(c)|\le A\sqrt{\log s}+L$,
one has
\begin{align}
        \E[(N-c+1)_+^a]&\le B_a s^{a/2}(1+u)^a,
        \label{eq:binom-moment-upper1}\\
        \E[(c-N)_+^a]&\le B_a s^{a/2}(1+u)^a.
\end{align}
Moreover, if $x_s(c)\le0$, then
\begin{equation}\label{eq:binom-moment-lower1}
        \E\!\big[(N-c+1)_+^a\,
        \1_{\{1\le N-c+1\le B_a(1+u)\sqrt{s}\}}\big]
        \ge b_a s^{a/2}(1+u)^a,
\end{equation}
and if $x_s(c)\ge0$, then
\begin{equation}\label{eq:binom-moment-lower2}
        \E\!\big[(c-N)_+^a\,
        \1_{\{1\le c-N\le B_a(1+u)\sqrt{s}\}}\big]
        \ge b_a s^{a/2}(1+u)^a.
\end{equation}
\end{lemma}

\begin{proof}
We prove \eqref{eq:binom-moment-upper1} and \eqref{eq:binom-moment-lower1}; the estimates involving $(c-N)_+$ follow by the same argument applied to $\mu_s-N$.

For the upper bound, use
\[
        (N-c+1)_+
        \le |N-\mu_s|+|\mu_s-c+1|.
\]
In the stated range,
\[
        |\mu_s-c+1|
        \le C(1+u)\sqrt{s}.
\]
Therefore, by $(r+t)^a\le C_a(r^a+t^a)$ and Lemma~\ref{lem:bernstein-binomial},
\[
        \E[(N-c+1)_+^a]
        \le C_a\E|N-\mu_s|^a+C_a(1+u)^a s^{a/2}
        \le B_a s^{a/2}(1+u)^a.
\]

Now assume $x_s(c)\le0$.  Then
\begin{equation}\label{eq:mean-excess-capacity}
        \mu_s-c+1
        =\sigma\sqrt{s}\,u+(1-q).
\end{equation}
We split into two cases.

First suppose $u\ge u_0$, where $u_0$ is a fixed constant chosen large enough below.  Let
\[
        E_u:=\left\{|N-\mu_s|\le \frac12\sigma\sqrt{s}\,u\right\}.
\]
Chebyshev's inequality gives
\[
        \Prob(E_u^c)
        \le \frac{4\Var(N)}{\sigma^2 s u^2}
        \le \frac{4}{u^2}.
\]
Choose $u_0$ so large that this probability is at most $1/2$.  On $E_u$, \eqref{eq:mean-excess-capacity} implies
\[
        N-c+1\ge \frac12\sigma\sqrt{s}\,u
\]
for all sufficiently large $s$, and also
\[
        N-c+1\le C\sqrt{s}\,u\le C(1+u)\sqrt{s}.
\]
Thus
\[
        \E\!\big[(N-c+1)_+^a\,
        \1_{\{1\le N-c+1\le C(1+u)\sqrt{s}\}}\big]
        \ge
        \frac12\left(\frac12\sigma\sqrt{s}\,u\right)^a
        \ge c s^{a/2}(1+u)^a,
\]
after reducing $c$ and using $u\ge u_0$.

It remains to consider $0\le u\le u_0$.  Let $I=[2,3]$.  Berry--Esseen gives
\[
        \Prob\left(\frac{N-\mu_s}{\sigma\sqrt{s}}\in I\right)
        \ge c_0>0
\]
for all large $s$, uniformly in $c$ because the interval is fixed and does not depend on $c$.  On this event,
\[
        N-c+1
        =(N-\mu_s)+(\mu_s-c+1)
        \ge 2\sigma\sqrt{s}+\sigma\sqrt{s}\,u+(1-q)
        \ge 2\sigma\sqrt{s}
        \ge c_1\sqrt{s},
\]
where the last inequality holds for all sufficiently large $s$ and uses $u\ge0$.  Moreover, using $u\le u_0$,
\[
        N-c+1\le C_1\sqrt{s}
        \le C_1(1+u)\sqrt{s}.
\]
Therefore the truncated expectation is at least $c_0(c_1\sqrt{s})^a$, which is comparable to $s^{a/2}(1+u)^a$ because $u$ is bounded.  Enlarging $B_a$ and reducing $b_a$ completes the proof of \eqref{eq:binom-moment-lower1}.
\end{proof}

The truncation in \eqref{eq:binom-moment-lower1} and \eqref{eq:binom-moment-lower2} ensures that the relevant conditional rank is at most $O((1+u)\sqrt{s})$, which is $o(s)$ in the moderate-deviation ranges used later.

\section{Edge quantiles and order statistics}\label{app:edge-order}\label{sec:edge-order}

This section introduces the local assumptions at the support edges adjacent to the gap, which we call gap-facing edges, and turns them into estimates for the offline marginal order statistic.  The results are deliberately stated in a reusable one-sided form and are applied at the two gap-facing support edges $a_2$ and $b_1$.

Fix a local mass exponent $\beta\ge0$ and define
\begin{equation}\label{eq:beta-exponents-model}
        \theta:=\frac1{\beta+1},
        \qquad
        \gamma:=\frac{\beta}{2(\beta+1)}=\frac{1-\theta}{2}.
\end{equation}

\begin{assumption}[Gap-facing edge exponent]\label{ass:model-gapped}
There exist constants $0<m_g<M_g<\infty$ and $\eps_g>0$ such that, for every $0<u\le\eps_g$,
\begin{align}
        m_g u^{\beta+1}
        &\le F_+(a_2+u)
        \le M_g u^{\beta+1},
        \label{eq:upper-edge-model}\\
        m_g u^{\beta+1}
        &\le 1-F_-(b_1-u)
        \le M_g u^{\beta+1}.
        \label{eq:lower-edge-model}
\end{align}
\end{assumption}

If densities exist, the gap-facing densities are $\asymp$ distance$^\beta$.  The case $\beta=0$ corresponds to bounded positive densities at the two gap-facing edges, while $\beta>0$ is the vanishing-density case used in Appendix~\ref{app:gapped-positive}.

\begin{lemma}[Gapped convex residual shape]\label{lem:gapped-residual-shape}
Under Assumption~\ref{ass:model-gapped}, there are constants $0<c<C<\infty$ such that, for all sufficiently small $u>0$,
\[
        c u^{\beta+2}
        \le
        \Delta(a_2+u)
        \le
        C u^{\beta+2},
        \qquad
        c u^{\beta+2}
        \le
        \Delta(b_1-u)
        \le
        C u^{\beta+2}.
\]
\end{lemma}

\begin{proof}
For $u>0$ small,
\[
        \Delta'(a_2+u)=qF_+(a_2+u),
        \qquad
        \Delta'(b_1-u)=-p\bigl(1-F_-(b_1-u)\bigr).
\]
Because $\Delta$ vanishes at both gap boundaries, integration gives
\[
        \Delta(a_2+u)
        =q\int_0^u F_+(a_2+v)\,\dd v,
        \qquad
        \Delta(b_1-u)
        =p\int_0^u \bigl(1-F_-(b_1-v)\bigr)\,\dd v.
\]
The two estimates then follow from \eqref{eq:upper-edge-model}--\eqref{eq:lower-edge-model}.  
\end{proof}

We use the following elementary convention.  A right-edge coordinate is denoted by $x\ge0$ and corresponds to values $e+x$ to the right of an edge $e$.  A left-edge coordinate is also denoted by $x\ge0$ and corresponds to values $e-x$ to the left of an edge $e$.

\subsection{One-sided quantile conversion}

Let $\mathcal A:[0,L]\to[0,1]$ be a continuous strictly increasing distribution function in edge-distance coordinates, with $\mathcal A(0)=0$ and $\mathcal A(L)=1$.  Its generalized inverse is
\[
        Q_{\mathcal A}(u):=\inf\{x\in[0,L]: \mathcal A(x)\ge u\},
        \qquad 0<u\le1.
\]
In applications, $\mathcal A(x)$ is either $F_+(a_2+x)$ or $1-F_-(b_1-x)$.

\begin{lemma}[Edge quantile bounds]\label{lem:edge-quantile-bounds}
Suppose that there are constants $m,M,\rho>0$ such that
\begin{equation}
        m x^{\beta+1}\le \mathcal A(x)\le Mx^{\beta+1},
        \qquad 0\le x\le \rho.
\end{equation}
Then there exist constants $0<c_Q<C_Q<\infty$ and $u_0\in(0,1)$ such that
\begin{equation}\label{eq:edge-quantile-power-local}
        c_Q u^\theta\le Q_{\mathcal A}(u)\le C_Q u^\theta,
        \qquad 0<u\le u_0.
\end{equation}
Moreover, after increasing $C_Q$ if necessary, one has the global upper bound
\begin{equation}\label{eq:edge-quantile-global-upper}
        Q_{\mathcal A}(u)\le C_Q u^\theta,
        \qquad 0<u\le1.
\end{equation}
\end{lemma}

\begin{proof}
Choose $u_0:=\min\{m\rho^{\beta+1},M\rho^{\beta+1},1/2\}$.  If $0<u\le u_0$, then both $(u/m)^\theta$ and $(u/M)^\theta$ are at most $\rho$.  The upper bound follows because
\[
        \mathcal A\bigl((u/m)^\theta\bigr)
        \ge m\bigl((u/m)^\theta\bigr)^{\beta+1}=u,
\]
so $Q_{\mathcal A}(u)\le m^{-\theta}u^\theta$.  The lower bound follows because if $x<(u/M)^\theta$, then $\mathcal A(x)\le Mx^{\beta+1}<u$, hence the inverse cannot be smaller than $(u/M)^\theta$.  This gives \eqref{eq:edge-quantile-power-local} with $c_Q=M^{-\theta}$ and $C_Q=m^{-\theta}$.

For $u\ge u_0$, compactness gives $Q_{\mathcal A}(u)\le L\le L u_0^{-\theta}u^\theta$.  Enlarging $C_Q$ to dominate both constants yields \eqref{eq:edge-quantile-global-upper}.
\end{proof}

The next lemma is the main conversion used in all source estimates.  It says that an edge-rank fraction $u$ produces convex residual cost of order at least $u^{1+\theta}$.

\begin{lemma}[Edge quantile-to-residual conversion]\label{lem:edge-conversion}
Let $\mathcal A$ satisfy the assumptions of Lemma~\ref{lem:edge-quantile-bounds}.  Let $\Delta_{\mathcal A}:[0,L]\to[0,\infty)$ be continuous with $\Delta_{\mathcal A}(0)=0$ and $\Delta_{\mathcal A}(x)>0$ for $x>0$.  Suppose that for some $m_\Delta,M_\Delta,\rho_\Delta>0$,
\begin{equation}\label{eq:edge-residual-power-local}
        m_\Delta x^{\beta+2}
        \le \Delta_{\mathcal A}(x)
        \le M_\Delta x^{\beta+2},
        \qquad 0\le x\le \rho_\Delta.
\end{equation}
Then there are constants $0<c_\Delta<C_\Delta<\infty$ such that, for all $0<u\le1$,
\begin{equation}\label{eq:edge-conversion}
        c_\Delta u^{1+\theta}
        \le
        \Delta_{\mathcal A}(Q_{\mathcal A}(u))
        \le
        C_\Delta u^{1+\theta}.
\end{equation}
The same conclusion applies at a left edge after writing $x=e-v$ for the distance from the edge.
\end{lemma}

\begin{proof}
Let $u_1>0$ be small enough that $Q_{\mathcal A}(u)\le \rho_\Delta$ for every $0<u\le u_1$, which is possible by Lemma~\ref{lem:edge-quantile-bounds}.  For $0<u\le u_1$, Lemma~\ref{lem:edge-quantile-bounds} and \eqref{eq:edge-residual-power-local} give
\[
        \Delta_{\mathcal A}(Q_{\mathcal A}(u))
        \ge m_\Delta (c_Q u^\theta)^{\beta+2}
        =m_\Delta c_Q^{\beta+2} u^{\theta(\beta+2)}.
\]
Since $\theta(\beta+2)=1+\theta$, this is the desired lower bound for small $u$.

For $u\in[u_1,1]$, the function $u\mapsto \Delta_{\mathcal A}(Q_{\mathcal A}(u))$ is positive.  More explicitly, $Q_{\mathcal A}(u)\ge Q_{\mathcal A}(u_1)>0$, and by the assumption $\Delta_{\mathcal A}(x)>0$ for $x>0$ and compactness of $[Q_{\mathcal A}(u_1),L]$,
\[
        m_1:=\inf_{x\in[Q_{\mathcal A}(u_1),L]}\Delta_{\mathcal A}(x)>0.
\]
Because $u^{1+\theta}\le1$, the same lower bound holds on $[u_1,1]$ with constant $m_1$.  Taking the smaller of the two constants proves the lower bound in \eqref{eq:edge-conversion}.

For the upper bound, use the local upper bounds for $u\le u_1$ and boundedness of $\Delta_{\mathcal A}$ for $u\ge u_1$, enlarging the constant if necessary by the same compactness argument.
\end{proof}

\begin{remark}[How the abstract conversion is used]\label{rem:edge-conversion-use}
For the upper gap edge $a_2$, take $\mathcal A(x)=F_+(a_2+x)$ and $\Delta_{\mathcal A}(x)=\Delta(a_2+x)$.  For the lower edge adjacent to the gap, $b_1$, take $\mathcal A(x)=1-F_-(b_1-x)$ and $\Delta_{\mathcal A}(x)=\Delta(b_1-x)$.  Assumption~\ref{ass:model-gapped} and Lemma~\ref{lem:gapped-residual-shape} imply the local power hypotheses.  The positivity condition $\Delta_{\mathcal A}(x)>0$ for $x>0$ follows because $h_0$ is a supporting affine function and $\Delta$ vanishes on the gap but is strictly increasing into either support edge for small positive edge distance.
\end{remark}

\subsection{Uniform order-statistic moments}

The conditional distribution of an edge order statistic is obtained by applying an edge quantile to a uniform order statistic.  We record the moment estimates needed for that conversion.

\begin{lemma}[Gamma-ratio bound]\label{lem:gamma-ratio-bound}
For every fixed $a>0$ there are constants $0<c_a<C_a<\infty$ such that
\begin{equation}
        c_a x^a
        \le
        \frac{\Gamma(x+a)}{\Gamma(x)}
        \le
        C_a x^a,
        \qquad x\ge1.
\end{equation}
\end{lemma}

\begin{proof}
The function
\[
        R_a(x):=\frac{\Gamma(x+a)}{\Gamma(x)x^a}
\]
is continuous and positive on $[1,\infty)$.  By Stirling's formula, $R_a(x)\to1$ as $x\to\infty$.  Therefore $R_a$ is bounded above and below by positive finite constants on $[1,\infty)$.
\end{proof}

\begin{lemma}[Uniform order-statistic moment]\label{lem:u-order-common}
Let $U_{j:n}$ be the $j$-th smallest order statistic among $n$ independent uniform $[0,1]$ random variables.  For every fixed $a>0$ there are constants $0<c_a<C_a<\infty$ such that, uniformly over $1\le j\le n$,
\begin{equation}\label{eq:u-order-moment-two-sided}
        c_a\left(\frac{j}{n+1}\right)^a
        \le
        \E[U_{j:n}^a]
        \le
        C_a\left(\frac{j}{n+1}\right)^a.
\end{equation}
\end{lemma}

\begin{proof}
The density of $U_{j:n}$ is beta with parameters $(j,n-j+1)$.  Hence
\begin{equation}
        \E[U_{j:n}^a]
        =\frac{\Gamma(j+a)}{\Gamma(j)}
         \frac{\Gamma(n+1)}{\Gamma(n+1+a)}.
\end{equation}
Applying Lemma~\ref{lem:gamma-ratio-bound} first at $x=j$ and then at $x=n+1$ gives
\[
        \E[U_{j:n}^a]
        \asymp_a
        \frac{j^a}{(n+1)^a},
\]
which is \eqref{eq:u-order-moment-two-sided}.
\end{proof}

\begin{corollary}[Convex residual moment of an edge order statistic]\label{cor:edge-order-residual}
Under the hypotheses of Lemma~\ref{lem:edge-conversion}, if $U_{j:n}$ is the $j$-th smallest uniform order statistic, then
\begin{equation}
        \E\bigl[\Delta_{\mathcal A}(Q_{\mathcal A}(U_{j:n}))\bigr]
        \ge
        c\left(\frac{j}{n+1}\right)^{1+\theta},
        \qquad 1\le j\le n,
\end{equation}
where $c>0$ depends only on the edge constants and on $\theta$. The upper bound follows from the upper bound in \eqref{eq:edge-conversion}.
\end{corollary}

\begin{proof}
By Lemma~\ref{lem:edge-conversion},
\[
        \Delta_{\mathcal A}(Q_{\mathcal A}(U_{j:n}))\ge c_\Delta U_{j:n}^{1+\theta}.
\]
Taking expectations and applying Lemma~\ref{lem:u-order-common} with $a=1+\theta$ proves the claim.  The upper estimate is identical.
\end{proof}

\subsection{Source estimates at a support gap}

We now combine the edge conversion with the binomial moment estimates from Section~\ref{sec:binomial}.  These estimates control the positive Jensen/order-statistic term for the gapped $\beta>0$ certificate.

Assume the two-support model of Assumption~\ref{ass:model-gapped}.  Let $N\sim\mathrm{Bin}(s-1,q)$ be the number of future arrivals in the upper support.  Conditional on $N=n$, the upper-support observations are independent draws from $F_+$ and the lower-support observations are independent draws from $F_-$.  Since every upper-support value exceeds every lower-support value, the offline marginal $H_{s,c}$ has the following conditional descriptions:
\begin{itemize}[leftmargin=2em]
\item If $n\ge c$, then $H_{s,c}$ lies in the upper support and is the upper-support order statistic of rank $n-c+1$.
\item If $n<c$, then $H_{s,c}$ lies in the lower support and is the lower-support order statistic of upper-tail rank $c-n$, equivalently the distance from the lower right edge $b_1$ has rank $c-n$ from the left.
\end{itemize}

\begin{lemma}[Gapped source estimate]\label{lem:gapped-random-rank-source}
Assume Assumption~\ref{ass:model-gapped}.  Fix constants $A,L<\infty$ with $A<1$.  There exist constants $0<c<C<\infty$ and $s_0<\infty$ such that, for all $s\ge s_0$ and all integers $c$ satisfying
\[
        u=|x_s(c)|\le A\sqrt{\log s}+L,
\]
one has
\begin{equation}\label{eq:gapped-source-bounds}
        c s^{\gamma-1}(1+u)^{1+\theta} \le \E\Delta(H_{s,c}) \le
        C s^{\gamma-1}(1+u)^{1+\theta}
        + C\,\Prob\bigl(|N-\mu_s|>C\sqrt{s\log s}\bigr).
\end{equation}
The constant inside the probability can be chosen large enough to make that probability $O(s^{-r})$ for any fixed $r<\infty$.
\end{lemma}

\begin{proof}
We first prove the lower bound when $x_s(c)\le0$; the case $x_s(c)\ge0$ is the same with the roles of the two support intervals interchanged.

Assume $x_s(c)\le0$.  Conditional on $N=n\ge c$, set $j=n-c+1$.  Then $H_{s,c}=a_2+Q_+(U_{j:n})$, where $Q_+$ is the lower-edge distance quantile of $F_+$ at $a_2$ and $U_{j:n}$ is the $j$-th smallest uniform order statistic among $n$ points.  Corollary~\ref{cor:edge-order-residual} gives
\begin{equation}
        \E\left[\Delta(H_{s,c})\mid N=n\right]
        \ge
        c_1\left(\frac{n-c+1}{n+1}\right)^{1+\theta},
        \qquad n\ge c.
\end{equation}
Since $n+1\le s$, the right side is at least $c_1((n-c+1)/s)^{1+\theta}$.  Therefore
\[
        \E\left[\Delta(H_{s,c})\1_{\{N\ge c\}}\right]
        \ge
        c_1\E\left[\left(\frac{N-c+1}{s}\right)^{1+\theta}\1_{\{N\ge c\}}\right].
\]
Dividing the truncated lower bound in Lemma~\ref{lem:binomial-moment-common} by $s^{1+\theta}$, with $a=1+\theta$, gives
\[
        \E\left[\Delta(H_{s,c})\1_{\{N\ge c\}}\right]
        \ge
        c_2s^{-(1+\theta)/2}(1+u)^{1+\theta}.
\]
Because $1+\theta=2(1-\gamma)$, $s^{-(1+\theta)/2}=s^{\gamma-1}$.  This gives the lower bound in \eqref{eq:gapped-source-bounds} when $x_s(c)\le0$.

For the upper bound, split according to $N\ge c$ and $N<c$.  On the event
\[
        \mathcal E_C:=\{|N-\mu_s|\le C\sqrt{s\log s}\},
\]
with $C$ fixed large enough, the displayed range for $c$ implies $c=qs+O(\sqrt{s\log s})$.  Hence, on $\mathcal E_C$ and for all large $s$, the relevant conditional sample sizes satisfy $n+1\ge(q/2)s$ on $\{N\ge c\}$ and $s-n\ge(p/2)s$ on $\{N<c\}$.  The corresponding ranks are
\[
        j_+(n)=n-c+1,
        \qquad
        j_-(n)=c-n,
\]
and both are bounded by $C'(1+u+\sqrt{\log s})\sqrt{s}$ on $\mathcal E_C$.  Using the upper half of Lemma~\ref{lem:edge-conversion} and Lemma~\ref{lem:u-order-common} conditionally on $N=n$ gives
\begin{align*}
        \E[\Delta(H_{s,c})\mid N=n]
        &\le C\left(\frac{j_+(n)}{s}\right)^{1+\theta},
        && n\ge c, \\
        \E[\Delta(H_{s,c})\mid N=n]
        &\le C\left(\frac{j_-(n)}{s}\right)^{1+\theta},
        && n<c.
\end{align*}
Taking expectations and applying Lemma~\ref{lem:binomial-moment-common} to $(N-c+1)_+$ and $(c-N)_+$ gives the term $Cs^{-(1+\theta)/2}(1+u)^{1+\theta}=C s^{\gamma-1}(1+u)^{1+\theta}$.  On $\mathcal E_C^c$, the convex residual is bounded because the value support is compact; this gives the probability term in \eqref{eq:gapped-source-bounds}.
\end{proof}

\section{A standard-scale finite-difference expansion}\label{app:fd-standard}\label{sec:fd-standard}
For a candidate certificate
\[
        B_s(c)=\eta a_sG_s(x_s(c)),
\]
recall the affine-transport drift
\[
        M_s(c)=(1-q)B_{s-1}(c)+qB_{s-1}(c-1)-B_s(c).
\]
We use the standard coordinate $x_s(c)=(c-qs)/(\sigma\sqrt{s})$ and allow the profile $G_s$ to depend on $s$.

\subsection{Predecessor-coordinate expansion}
Throughout this subsection, let $\xi$ be the two-point random variable
\[
        \xi=\begin{cases}
        q, & \text{with probability }1-q,\\
        -(1-q), & \text{with probability }q.
        \end{cases}
\]
Then $\E\xi=0$, $\E\xi^2=\sigma^2$, and $|\xi|\le1$.

\begin{lemma}[Predecessor coordinates]\label{lem:predecessor-coordinates}
Fix $R_s\ge1$ with $R_s=O(\sqrt{\log s})$.  Uniformly over states with $|x_s(c)|\le R_s$, writing $x=x_s(c)$, define
\[
        X_s:=\frac{\sigma\sqrt{s}\,x+\xi}{\sigma\sqrt{s-1}}.
\]
Then $X_s=x_{s-1}(c)$ with probability $1-q$ and $X_s=x_{s-1}(c-1)$ with probability $q$.  Moreover,
\begin{equation}\label{eq:Xs-expanded}
        X_s=x+\frac{\xi}{\sigma\sqrt{s}}+\frac{x}{2s}+R_s(x,\xi),
        \qquad
        |R_s(x,\xi)|\le C\frac{1+|x|}{s^{3/2}}.
\end{equation}
\end{lemma}

\begin{proof}
By definition,
\[
        x_{s-1}(c)=\frac{c-q(s-1)}{\sigma\sqrt{s-1}}
        =\frac{\sigma\sqrt{s}\,x+q}{\sigma\sqrt{s-1}},
\]
and similarly
\[
        x_{s-1}(c-1)=\frac{\sigma\sqrt{s}\,x-(1-q)}{\sigma\sqrt{s-1}}.
\]
Since
\[
        \sqrt{\frac{s}{s-1}}=1+\frac1{2s}+O(s^{-2}),
        \qquad
        \frac1{\sqrt{s-1}}=\frac1{\sqrt{s}}\left(1+O(s^{-1})\right),
\]
Substitution gives the stated representation of $X_s$ and the error bound.  The identities for the two predecessor coordinates follow from taking $\xi=q$ and $\xi=-(1-q)$, respectively.  
\end{proof}

\subsection{Taylor expansion of the affine transport}

For an interval $I\subset\mathbb R$ write
\[
        \|f\|_{j,I}:=\sup_{y\in I}|f^{(j)}(y)|.
\]
For a state $(s,c)$ and a local radius $\rho_{\rm loc}>0$, define
\[
        I_{s,c}(\rho_{\rm loc}):=[x_s(c)-\rho_{\rm loc},x_s(c)+\rho_{\rm loc}].
\]

\begin{lemma}[Standard-scale Bellman drift expansion]\label{lem:standard-drift}
Let
\begin{equation}
        B_s(c)=\eta a_sG_s(x_s(c))
\end{equation}
for positive numbers $a_s$ and functions $G_s\in C^3(\mathbb R)$.  Fix $R_s=O(\sqrt{\log s})$ and suppose $|x_s(c)|\le R_s$.  Let $x=x_s(c)$ and let $I=I_{s,c}(1)$.  Then, for all sufficiently large $s$,
\begin{align}
        M_s(c)
        &=(1-q)B_{s-1}(c)+qB_{s-1}(c-1)-B_s(c)\notag\\
        &=\eta a_s\left[\frac{G_s''(x)}{2s}+\frac{xG_s'(x)}{2s}\right]
          +\eta(a_{s-1}-a_s)G_s(x)
          +\eta a_s\bigl(G_{s-1}(x)-G_s(x)\bigr)
          +\mathrm{Err}_{s,c},\label{eq:standard-drift-expansion}
\end{align}
where the error satisfies
\begin{align}
        |\mathrm{Err}_{s,c}|
        \le C\eta a_s\Bigg[&
        \frac{1+x^2}{s^2}\|G_s''\|_{0,I}
        +\frac1{s^{3/2}}\|G_s'''\|_{0,I}
        +\frac{1+|x|}{s^{3/2}}\|G_s'\|_{0,I}
        \notag\\
        &+\frac1s\|G_{s-1}''-G_s''\|_{0,I}
        +\frac{1+|x|}{s}\|G_{s-1}'-G_s'\|_{0,I}
        \notag\\
        &+\left|\frac{a_{s-1}}{a_s}-1\right|
        \|G_{s-1}-G_s\|_{0,I}
        +\left|\frac{a_{s-1}}{a_s}-1\right|
        \frac{|x|}{s}\|G_s'\|_{0,I}
        +\left|\frac{a_{s-1}}{a_s}-1\right|
        \frac1s\|G_s''\|_{0,I}
        \Bigg].\label{eq:standard-drift-error}
\end{align}
The constant $C$ depends only on $q$ and on the implicit constant in $R_s=O(\sqrt{\log s})$.
\end{lemma}

\begin{proof}
Let $\xi$ and $X_s$ be as in Lemma~\ref{lem:predecessor-coordinates}.  Then
\begin{equation}\label{eq:transport-as-expectation}
        (1-q)G_{s-1}(x_{s-1}(c))+qG_{s-1}(x_{s-1}(c-1))
        =\E G_{s-1}(X_s).
\end{equation}
By \eqref{eq:Xs-expanded}, $|X_s-x|\le C(1+R_s)/\sqrt{s}=o(1)$ uniformly on the active range; hence $X_s\in I$ for all sufficiently large $s$.  In the Taylor expansions below, the Lagrange intermediate point lies between $x$ and $X_s$ and therefore also lies in $I$.

First expand $G_s$ around $x$.  Taylor's theorem with remainder gives
\[
        \E G_s(X_s)
        =G_s(x)+G_s'(x)\E(X_s-x)
        +\frac12G_s''(x)\E(X_s-x)^2+R_3,
\]
where
\[
        |R_3|\le C\|G_s'''\|_{0,I}\E|X_s-x|^3.
\]
From Lemma~\ref{lem:predecessor-coordinates},
\begin{align}
        \E(X_s-x)&=\frac{x}{2s}+O\left(\frac{1+|x|}{s^{3/2}}\right),\label{eq:Xs-first-moment-fd}\\
        \E(X_s-x)^2&=\frac1s+O\left(\frac{1+x^2}{s^2}\right),\label{eq:Xs-second-moment-fd}\\
        \E|X_s-x|^3&\le Cs^{-3/2}.
\end{align}
Here the cross terms involving the remainder $R_s(x,\xi)$ are not zero in general, because $R_s$ depends on $\xi$; they are bounded directly by $\E|R_s(x,\xi)|\le C(1+|x|)s^{-3/2}$ and $\E|\xi R_s(x,\xi)|\le C(1+|x|)s^{-3/2}$, which is precisely the size absorbed in \eqref{eq:Xs-first-moment-fd}--\eqref{eq:Xs-second-moment-fd}.
Therefore
\begin{align}
        \E G_s(X_s)
        &=G_s(x)+\frac{xG_s'(x)}{2s}+\frac{G_s''(x)}{2s}+R_s^{(1)},\label{eq:EGs-expansion}
\end{align}
with
\begin{equation}
        |R_s^{(1)}|
        \le C\left[
        \frac{1+|x|}{s^{3/2}}\|G_s'\|_{0,I}
        +\frac{1+x^2}{s^2}\|G_s''\|_{0,I}
        +\frac1{s^{3/2}}\|G_s'''\|_{0,I}
        \right].
\end{equation}

Next compare $G_{s-1}$ with $G_s$.  By Taylor expanding the difference $G_{s-1}-G_s$ to second order and using \eqref{eq:Xs-first-moment-fd}--\eqref{eq:Xs-second-moment-fd},
\begin{align}
        \left|\E\bigl[(G_{s-1}-G_s)(X_s)\bigr]-(G_{s-1}-G_s)(x)\right|
        \le C\left[
        \frac{1+|x|}{s}\|G_{s-1}'-G_s'\|_{0,I}
        +\frac1s\|G_{s-1}''-G_s''\|_{0,I}
        \right].\label{eq:Gsminus-time-difference-bound}
\end{align}
Combining \eqref{eq:transport-as-expectation}, \eqref{eq:EGs-expansion}, and \eqref{eq:Gsminus-time-difference-bound},
\begin{align}
        \E G_{s-1}(X_s)
        &=G_s(x)+\frac{xG_s'(x)}{2s}+\frac{G_s''(x)}{2s}
          +(G_{s-1}(x)-G_s(x))
          +R_s^{(2)},\label{eq:EGsminus-expansion}
\end{align}
where $R_s^{(2)}$ is bounded by the first five terms in \eqref{eq:standard-drift-error} without the factor $\eta a_s$.

Finally,
\[
        M_s(c)=\eta a_{s-1}\E G_{s-1}(X_s)-\eta a_sG_s(x).
\]
Insert \eqref{eq:EGsminus-expansion}, write $a_{s-1}=a_s+(a_{s-1}-a_s)$, and keep the leading contribution of $(a_{s-1}-a_s)G_s(x)$.  The product of $(a_{s-1}-a_s)$ with the time-change term $G_{s-1}(x)-G_s(x)$ is the cross term
\[
        \eta(a_{s-1}-a_s)\bigl(G_{s-1}(x)-G_s(x)\bigr),
\]
which is bounded by the new term
\[
        C\eta a_s\left|\frac{a_{s-1}}{a_s}-1\right|\|G_{s-1}-G_s\|_{0,I}
\]
in \eqref{eq:standard-drift-error}.  The multiplication of $(a_{s-1}/a_s-1)$ against the transport derivative terms gives the final two terms in \eqref{eq:standard-drift-error}.  This proves \eqref{eq:standard-drift-expansion}--\eqref{eq:standard-drift-error}.
\end{proof}

\subsection{Useful corollaries}

The following two consequences are the forms used later in the proof.

\begin{corollary}[Power-height drift]\label{cor:power-height-drift}
Let $a_s=s^\gamma$ with fixed $\gamma\in[0,1/2)$.  Assume that on the active range $|x_s(c)|\le R_s=O(\sqrt{\log s})$ the profile sequence satisfies
\begin{align*}
        |G_s^{(j)}(y)|&\le A_sK_s(1+|y|)^{m_j},\qquad j=0,1,2,3,\\
        |G_{s-1}^{(j)}(y)-G_s^{(j)}(y)|&\le \frac{A_sK_s}{s\log s}(1+|y|)^{m_j},
        \qquad j=0,1,2,
\end{align*}
for deterministic scales $A_s,K_s$ growing at most polynomially in $\log s$.  Then
\begin{equation}\label{eq:power-height-drift-cor}
        M_s(c)
        =\eta s^{\gamma-1}
        \left[\frac12G_s''(x)+\frac12xG_s'(x)-\gamma G_s(x)\right]
        +\mathrm{Rem}_{s,c},
\end{equation}
where $|\mathrm{Rem}_{s,c}|$ is bounded by
\[
        \eta s^\gamma\|G_{s-1}-G_s\|_{0,I}
        +\text{the right side of \eqref{eq:standard-drift-error} with }a_s=s^\gamma,
\]
including the cross term involving $G_{s-1}-G_s$.  In particular, the remainder is lower order whenever the displayed derivative and time-variation bounds are lower order relative to the intended source term.
\end{corollary}

\begin{proof}
For $a_s=s^\gamma$,
\[
        a_{s-1}-a_s=-\gamma s^{\gamma-1}+O(s^{\gamma-2}).
\]
Substituting this into Lemma~\ref{lem:standard-drift} gives \eqref{eq:power-height-drift-cor}.  The explicit time-change term $\eta s^\gamma(G_{s-1}(x)-G_s(x))$, the $O(s^{\gamma-2})G_s(x)$ part of $a_{s-1}-a_s$, and the error bound \eqref{eq:standard-drift-error} are all included in $\mathrm{Rem}_{s,c}$.
\end{proof}

\begin{corollary}[First difference on the standard scale]\label{cor:standard-first-difference}
Let $B_s(c)=\eta a_sG_s(x_s(c))$.  Suppose $G_s$ is continuously differentiable on the active range and $|x_s(c)|\le R_s=O(\sqrt{\log s})$.  Then, whenever the line segment between $x_{s-1}(c)$ and $x_{s-1}(c-1)$ lies inside the active range enlarged by one,
\begin{equation}\label{eq:first-difference-standard}
        |D_s(c)|
        =|B_{s-1}(c)-B_{s-1}(c-1)|
        \le
        \frac{C\eta a_{s-1}}{\sqrt{s}}
        \sup_{y\in I_{s,c}(1)}|G_{s-1}'(y)|.
\end{equation}
\end{corollary}

\begin{proof}
The predecessor coordinates differ by
\[
        x_{s-1}(c)-x_{s-1}(c-1)=\frac1{\sigma\sqrt{s-1}}.
\]
The mean-value theorem gives
\[
        |D_s(c)|
        \le
        \eta a_{s-1}
        \frac1{\sigma\sqrt{s-1}}
        \sup_{y\in I_{s,c}(1)}|G_{s-1}'(y)|,
\]
which is \eqref{eq:first-difference-standard} after increasing the constant.
\end{proof}

\section{The gapped vanishing-density lower bound, $\beta>0$}\label{app:gapped-positive}\label{sec:gapped-positive}

This section proves the gapped positive-exponent extension, where the support has a gap and the gap-facing densities vanish with exponent $\beta>0$.  The proof uses the common Bellman certificate inequality from Section~\ref{sec:framework}, the binomial estimates from Appendix~\ref{app:binomial}, the edge-order estimates from Appendix~\ref{app:edge-order}, and the standard-scale finite-difference expansion from Appendix~\ref{app:fd-standard}.  The new ingredient in this section is the form of the certificate.  It has a harmonic core on the $\sqrt{s}$ scale, but the core is allowed to grow out to the moderate-deviation radius $\sqrt{s\log s}$.

Throughout this section Assumption~\ref{ass:model-gapped} holds with $\beta>0$ for the harmonic-certificate construction; the exponents $\theta$ and $\gamma$ are those in \eqref{eq:beta-exponents-model}.  The mean boundary overshoot and gap-margin estimates in Subsection~\ref{subsec:gapped-source-margin} are stated in a form that also covers the case $\beta=0$.  Recall that $\Delta\equiv0$ on the gap $[b_1,a_2]$, that $\Delta\ge0$ is convex, and that near the two gap-facing edges
\begin{equation}
        \Delta(a_2+u)\asymp u^{\beta+2},
        \qquad
        \Delta(b_1-u)\asymp u^{\beta+2}.
\end{equation}

The lower bound is proved by constructing a feasible Bellman certificate supported in the moderate-deviation band
\[
        |c-qs|\lesssim \sqrt{s\log s}.
\]
At sufficiently small fixed shifts on this scale, the harmonic core contributes the height needed for the lower bound in Theorem~\ref{thm:gapped-shift-positive}.

\subsection{The harmonic core profile}\label{subsec:gapped-harmonic-profile}

The drift operator associated with standard-scale Bellman transport is
\begin{equation}
        \calL_\gamma f(x):=\frac12 f''(x)+\frac12 x f'(x)-\gamma f(x).
\end{equation}
The profile used near the critical line must be nearly annihilated by this operator.  Otherwise the negative drift near $x=O(1)$ would be too large to be balanced by the source.
Define
\[
        H_\gamma(x):=\E |x+Z|^{2\gamma},
        \qquad Z\sim N(0,1),\quad x\in\R.
\]
\begin{lemma}[Harmonic core]\label{lem:Hgamma-properties}
The function $H_\gamma$ belongs to $C^\infty(\R)$, is even and strictly positive, and
\begin{equation}\label{eq:Hgamma-ODE}
        \calL_\gamma H_\gamma=0.
\end{equation}
For each integer $j\ge0$ there is $C_j<\infty$ such that
\[
        |H_\gamma^{(j)}(x)|\le C_j(1+|x|)^{2\gamma},
        \qquad x\in\R.
\]
For every integer $j\ge1$,
\begin{equation}\label{eq:Hgamma-derivatives-uniform}
        \sup_{x\in\R}|H_\gamma^{(j)}(x)|<\infty.
\end{equation}
Finally, there are constants $0<c_0<C_0<\infty$ such that, for $|x|\ge1$,
\begin{equation}\label{eq:Hgamma-sharp-growth}
        c_0 |x|^{2\gamma}\le H_\gamma(x)\le C_0 |x|^{2\gamma}.
\end{equation}
\end{lemma}
\begin{proof}
Write
\[
        H_\gamma(x)=\int_{\R}|y|^{2\gamma}\varphi(y-x)\,\dd y,
\]
where $\varphi$ is the standard normal density.  This is the convolution of the locally integrable function $y\mapsto |y|^{2\gamma}$ with a smooth Gaussian density.  Hence $H_\gamma\in C^\infty(\R)$, and differentiation may be performed on the Gaussian factor.  Evenness and positivity are immediate.

We next prove the differential equation.  The singularity of $|y|^{2\gamma}$ at zero is harmless, but it is better not to differentiate the singular integrand twice.  Let $f_\eps(y)=(y^2+\eps^2)^\gamma$ and $H_{\gamma,\eps}(x):=\E f_\eps(x+Z)$.  For each fixed $\eps>0$, Stein's identity gives, with $W=x+Z$,
\[
        H_{\gamma,\eps}''(x)+xH_{\gamma,\eps}'(x)
        =\E[W f_\eps'(W)].
\]
Here $Wf_\eps'(W)=2\gamma W^2(W^2+\eps^2)^{\gamma-1}$ is bounded by a constant times $1+|W|^{2\gamma}$ uniformly in $\eps\le1$.  Hence dominated convergence, uniformly on compact $x$-sets, yields
\[
        \lim_{\eps\downarrow0}\E[Wf_\eps'(W)]
        =2\gamma\E|W|^{2\gamma}=2\gamma H_\gamma(x).
\]
The convergence $H_{\gamma,\eps}\to H_\gamma$ and the convergence of the derivatives follow by differentiating the Gaussian kernel in the convolution representation above.  Therefore
\[
        H_\gamma''(x)+xH_\gamma'(x)=2\gamma H_\gamma(x),
\]
which is equivalent to \eqref{eq:Hgamma-ODE}.

The global derivative bounds follow by differentiating the Gaussian kernel.  The derivative $H_\gamma^{(j)}(x)$ is an integral of $|y|^{2\gamma}$ times a degree-$j$ polynomial in $(y-x)$ times $\varphi(y-x)$; after the change of variables $z=y-x$, this is bounded by
\[
        C_j\E(1+|x+Z|)^{2\gamma}(1+|Z|)^j
        \le C_j'(1+|x|)^{2\gamma}.
\]
For $|x|\ge1$, the upper bound in \eqref{eq:Hgamma-sharp-growth} and the derivative bounds
\[
        |H_\gamma^{(j)}(x)|\le C_j |x|^{2\gamma-j},
        \qquad |x|\ge1,\ j\ge1,
\]
follow similarly from Taylor expansion on the event $|Z|\le |x|/2$ and from Gaussian tails on the complement.  The lower bound for $H_\gamma$ follows from
\[
        H_\gamma(x)\ge \E\bigl[|x+Z|^{2\gamma}\1_{\{|Z|\le |x|/2\}}\bigr]
        \ge (|x|/2)^{2\gamma}\Prob(|Z|\le |x|/2),
\]
which is at least a positive constant times $|x|^{2\gamma}$ for $|x|\ge1$.

For $j\ge1$, the large-$x$ derivative bound above and the inequality $0<2\gamma<1$ imply boundedness on $|x|\ge1$.  On $|x|\le1$, continuity gives boundedness.  This proves \eqref{eq:Hgamma-derivatives-uniform}.
\end{proof}

Choose constants
\begin{equation}
        0<\alpha_1<\alpha_2<A<\sqrt\theta.
\end{equation}
Let $\psi\in C^\infty(\R)$ be even, with $0\le\psi\le1$, such that
\begin{equation}
        \psi(u)=1\quad (|u|\le\alpha_1),
        \qquad
        \psi(u)=0\quad (|u|\ge\alpha_2).
\end{equation}
For all sufficiently large $s$, define
\begin{equation}
        G_s(x):=(\log s)^\theta H_\gamma(x)\psi(x/\sqrt{\log s}).
\end{equation}
The certificate is
\begin{equation}\label{eq:gapped-B-def-pub}
        B_s(c):=\eta s^\gamma G_s(x_s(c)),
        \qquad 1\le c<s,
\end{equation}
with boundary values $B_s(0)=0$ and $B_s(c)=0$ for $c\ge s$.  Here $\eta>0$ is a small constant chosen only after all other constants and the base time are fixed.

\begin{lemma}[Profile bounds]\label{lem:gapped-profile-pub}
There is a constant $C<\infty$ such that, for all sufficiently large $s$, the following bounds hold.
\begin{enumerate}[label=(\roman*), leftmargin=2em]
\item The support of $G_s$ is contained in $[-\alpha_2\sqrt{\log s},\alpha_2\sqrt{\log s}]$, and
\begin{equation}
        |G_s(x)|\le C (\log s)^\theta(1+|x|)^{2\gamma}\1_{\{|x|\le\alpha_2\sqrt{\log s}\}}.
\end{equation}
In particular $|G_s(x)|\le C (\log s)^{(1+\theta)/2}$ everywhere.
\item For $j=1,2,3$,
\begin{equation}\label{eq:gapped-profile-global-derivative}
        |G_s^{(j)}(x)|\le C (\log s)^\theta\1_{\{|x|\le\alpha_2\sqrt{\log s}\}}.
\end{equation}
\item On the collar $\alpha_1\sqrt{\log s}\le |x|\le\alpha_2\sqrt{\log s}$,
\begin{align}\label{eq:gapped-profile-collar}
        |G_s(x)|&\le C(\log s)^{(1+\theta)/2},
        & |G_s'(x)|&\le C(\log s)^{\theta/2}, \notag\\
        |G_s''(x)|&\le C(\log s)^{(\theta-1)/2},
        & |G_s'''(x)|&\le C(\log s)^{(\theta-2)/2}.
\end{align}
\item Uniformly for $|x|\le\alpha_2\sqrt{\log s}+2$ and $j=0,1,2$,
\begin{equation}\label{eq:gapped-profile-time-derivatives}
        |G_{s-1}^{(j)}(x)-G_s^{(j)}(x)|
        \le \frac{C}{s}(\log s)^{(1+\theta)/2}.
\end{equation}
In the core region $|x|\le\alpha_1\sqrt{\log s}-1$, the sharper bound
\begin{equation}\label{eq:gapped-profile-time-core}
        |G_{s-1}(x)-G_s(x)|
        \le \frac{C}{s}(\log s)^{\theta-1}(1+|x|)^{2\gamma}
\end{equation}
holds.
\end{enumerate}
\end{lemma}

\begin{proof}
The support statement follows from the support of $\psi$.  The size bound follows from Lemma~\ref{lem:Hgamma-properties}:
\[
        |G_s(x)|\le C(\log s)^\theta(1+|x|)^{2\gamma}\1_{\{|x|\le\alpha_2\sqrt{\log s}\}}.
\]
On the support, $(1+|x|)^{2\gamma}\le C (\log s)^\gamma$, so $|G_s|\le C (\log s)^{\theta+\gamma}=C(\log s)^{(1+\theta)/2}$, because $\theta+\gamma=(1+\theta)/2$.

For derivatives, differentiate the product $(\log s)^\theta H_\gamma(x)\psi(x/\sqrt{\log s})$.  Every derivative falling on $\psi(x/\sqrt{\log s})$ produces a factor $(\log s)^{-1/2}$, and the derivatives of $\psi$ are bounded.  The conservative global bound \eqref{eq:gapped-profile-global-derivative} follows from the uniform derivative bound \eqref{eq:Hgamma-derivatives-uniform} for the terms in which at least one derivative falls on $H_\gamma$, and from the size bound $|H_\gamma(x)|\le C(1+|x|)^{2\gamma}$ for the terms in which all derivatives fall on the cutoff.  For example,
\[
        \left|(\log s)^\theta H_\gamma(x)\frac{1}{\sqrt{\log s}}\psi'(x/\sqrt{\log s})\right|
        \le C (\log s)^\theta(\log s)^{\gamma-1/2}=C(\log s)^{\theta/2}\le C(\log s)^\theta.
\]
In the collar $|x|\asymp \sqrt{\log s}$, the sharper estimates in \eqref{eq:Hgamma-sharp-growth} give
\[
        (\log s)^\theta|H_\gamma(x)|=O((\log s)^{(1+\theta)/2}),
        \qquad
        (\log s)^\theta|H_\gamma'(x)|=O((\log s)^{\theta/2}),
\]
and the cutoff derivative term in $G_s'$ is also $O((\log s)^\theta(\log s)^\gamma(\log s)^{-1/2})=O((\log s)^{\theta/2})$.  The product rule gives, for instance,
\[
G_s''=(\log s)^\theta\left[H_\gamma''\psi+2H_\gamma'(\log s)^{-1/2}\psi'+H_\gamma (\log s)^{-1}\psi''\right](x/\sqrt{\log s}),
\]
whose three terms are all $O((\log s)^{(\theta-1)/2})$ on the collar; differentiating once more gives terms of order $O((\log s)^{(\theta-2)/2})$.  This proves \eqref{eq:gapped-profile-collar}.

For time differences, write $G_s(x)=\Gamma(\sqrt{\log s},x)$ with
\[
        \Gamma(Y,x):=Y^{2\theta}H_\gamma(x)\psi(x/Y).
\]
For $|x|\le\alpha_2\sqrt{\log s}+2$ and $Y$ between $\sqrt{\log(s-1)}$ and $\sqrt{\log s}$, the same product-rule bounds used above give
\[
        |\partial_Y\partial_x^j\Gamma(Y,x)|\le C (\log s)^{\theta/2},
        \qquad j=0,1,2.
\]
We carry out the largest case explicitly.  Differentiating the displayed formula for $\partial_x^2\Gamma$ in $Y$ produces terms such as
\[
        Y^{2\theta-1}H_\gamma''\psi,
        \quad
        Y^{2\theta-2}H_\gamma'\psi',
        \quad
        Y^{2\theta-3}H_\gamma\psi'',
        \quad
        xY^{2\theta-4}H_\gamma\psi''',
\]
all evaluated at $x/Y$.  In the present range, the sharp bounds for $H_\gamma$ on the collar and the uniform derivative bounds in the core make each of these at most $C(\log s)^{\theta/2}$; the cases $j=0,1$ are easier.  Since $\sqrt{\log s}-\sqrt{\log(s-1)}=O((s\sqrt{\log s})^{-1})$, this implies the stronger bound $O(s^{-1}(\log s)^{(\theta-1)/2})$; the displayed weaker bound \eqref{eq:gapped-profile-time-derivatives} follows.  In the core, $\psi\equiv1$ for both $\sqrt{\log s}$ and $\sqrt{\log(s-1)}$, so $\Gamma(Y,x)=Y^{2\theta}H_\gamma(x)$.  The mean-value theorem gives
\[
        |G_{s-1}(x)-G_s(x)|
        \le C|\sqrt{\log s}-\sqrt{\log(s-1)}|(\log s)^{\theta-1/2}H_\gamma(x)
        \le \frac{C}{s}(\log s)^{\theta-1}(1+|x|)^{2\gamma},
\]
which is \eqref{eq:gapped-profile-time-core}.
\end{proof}

\subsection{Source estimates and gap margin}\label{subsec:gapped-source-margin}

The certificate is supported on $|x_s(c)|\le\alpha_2\sqrt{\log s}$.  Since $\alpha_2<\sqrt\theta$, this active band lies in a moderate-deviation range where the offline marginal mean remains in the support gap with a polynomially large margin.  We first record the source estimate, then prove the margin and show that the certificate perturbation is too small to move the threshold out of the gap.

The mean boundary overshoot and gap-margin estimates are common estimates for the gapped regime.  Although this section proves the positive-$\beta$ certificate, Lemmas~\ref{lem:gapped-mean-overshoot} and~\ref{lem:gapped-gap-margin-pub} are stated and proved under Assumption~\ref{ass:model-gapped} for any $\beta\ge0$, with $\theta=1/(\beta+1)$.  Their moderate-deviation radius is denoted by $A_0$ in the statements below; in the harmonic-certificate argument one takes $A_0=A$, while in Section~\ref{sec:gapped-zero} one may take any fixed $A_0<1$ because $\theta=1$.

\begin{lemma}[Gapped source on the active band]\label{lem:gapped-source-pub}
There are constants $c_S>0$ and $s_S<\infty$ such that, for all $s\ge s_S$ and all $1\le c<s$ satisfying
\begin{equation}\label{eq:gapped-active-source-condition}
        |x_s(c)|\le A \sqrt{\log s},
\end{equation}
one has
\begin{equation}\label{eq:gapped-source-pub}
        \E\Delta(H_{s,c})
        \ge c_S s^{\gamma-1}(1+|x_s(c)|)^{1+\theta}.
\end{equation}
\end{lemma}

\begin{proof}
This is a direct consequence of Lemma~\ref{lem:gapped-random-rank-source}.  The condition \eqref{eq:gapped-active-source-condition} is the same as $|x_s(c)|\le A\sqrt{\log s}$, with $A<1$ because $A<\sqrt\theta\le1$.  The two edge assumptions required in Lemma~\ref{lem:gapped-random-rank-source} are exactly the edge-mass assumptions in Assumption~\ref{ass:model-gapped}.  Lemma~\ref{lem:gapped-random-rank-source} gives the lower bound in \eqref{eq:gapped-source-bounds}, and $1+\theta=2(1-\gamma)$ gives precisely \eqref{eq:gapped-source-pub}.
\end{proof}

The next estimate controls the average amount by which the offline marginal overshoots into one of the two support intervals.  It is used to locate $\tau_s(c)=\E H_{s,c}$ inside the gap.

\begin{lemma}[Mean boundary overshoot]\label{lem:gapped-mean-overshoot}
Under Assumption~\ref{ass:model-gapped} with any $\beta\ge0$, let $\theta$ be as in \eqref{eq:beta-exponents-model} and fix $A_0<\infty$.  There are constants $C_P<\infty$ and $s_P<\infty$, depending also on $A_0$, such that, for all $s\ge s_P$ and all $c$ satisfying $|x_s(c)|\le A_0\sqrt{\log s}$, with $u:=|x_s(c)|$,
\begin{align}
        \E[(H_{s,c}-a_2)_+]
        &\le C_P s^{-\theta/2}(1+u)^\theta,
        \label{eq:upper-mean-overshoot}\\
        \E[(b_1-H_{s,c})_+]
        &\le C_P s^{-\theta/2}(1+u)^\theta.
\end{align}
\end{lemma}

\begin{proof}
We prove \eqref{eq:upper-mean-overshoot}; the other estimate is symmetric.  The random variable $(H_{s,c}-a_2)_+$ is nonzero only on $\{N\ge c\}$.  Conditional on $N=n\ge c$, $H_{s,c}$ is the upper-support order statistic of rank $n-c+1$.  The lower-edge quantile bound for the upper support applies because Assumption~\ref{ass:model-gapped} gives edge mass comparable to distance$^{\beta+1}$ at $a_2$.  Thus, by the global upper bound in Lemma~\ref{lem:edge-quantile-bounds}, after increasing the constant as in that lemma, the distance of this order statistic from $a_2$ is at most $C_Q U_{j:n}^{\theta}$ in expectation, where $j=n-c+1$.  Lemma~\ref{lem:u-order-common} therefore gives
\[
        \E[(H_{s,c}-a_2)_+\mid N=n]
        \le C\left(\frac{n-c+1}{n+1}\right)^\theta,
        \qquad n\ge c.
\]
On the event $n\ge c$, and for all states with $|x_s(c)|\le A_0\sqrt{\log s}$, we have $n+1\ge c\ge(q/2)s$ for all large $s$, because $c=qs+O(\sqrt{s\log s})$.  Hence
\[
        \E[(H_{s,c}-a_2)_+]
        \le C s^{-\theta}\E[(N-c+1)_+^\theta].
\]
The positive-part moment bound of Lemma~\ref{lem:binomial-moment-common}, applied with exponent $\theta$, gives
\[
        \E[(N-c+1)_+^\theta]
        \le C s^{\theta/2}(1+u)^\theta,
\]
which proves \eqref{eq:upper-mean-overshoot}.
\end{proof}

\begin{lemma}[Gap margin for the offline marginal]\label{lem:gapped-gap-margin-pub}
Under Assumption~\ref{ass:model-gapped} with any $\beta\ge0$, let $\theta=1/(\beta+1)$ and fix $A_0<\sqrt\theta$.  There are constants $c_G>0$ and $s_G<\infty$, depending also on $A_0$, such that, for all $s\ge s_G$ and all $c$ satisfying $|x_s(c)|\le A_0\sqrt{\log s}$,
\begin{equation}\label{eq:gapped-gap-margin-pub}
        \dist\bigl(\tau_s(c),\{b_1,a_2\}\bigr)
        \ge c_G s^{-A_0^2/2}(\log s)^{-1/2}.
\end{equation}
\end{lemma}

\begin{proof}
We prove the claim for the case $x_s(c)\le0$.  Put $u=|x_s(c)|$.  Since $H_{s,c}\le b_1$ on $\{N<c\}$ and $H_{s,c}\ge a_2$ on $\{N\ge c\}$,
\begin{equation}\label{eq:a2-margin-decomp}
        a_2-\tau_s(c)
        =\E[(a_2-H_{s,c})\1_{\{N<c\}}]
         -\E[(H_{s,c}-a_2)\1_{\{N\ge c\}}]
        \ge G\Prob(N<c)-\E[(H_{s,c}-a_2)_+].
\end{equation}
By Lemma~\ref{lem:binom-tail-common},
\begin{equation}
        \Prob(N<c)\ge c_0\frac{e^{-u^2/2}}{1+u}
        \ge c_0 s^{-A_0^2/2}(\log s)^{-1/2},
\end{equation}
where the final inequality uses $u\le A_0\sqrt{\log s}$.  By Lemma~\ref{lem:gapped-mean-overshoot},
\begin{equation}\label{eq:gap-margin-overshoot-small}
        \E[(H_{s,c}-a_2)_+]
        \le C s^{-\theta/2}(1+u)^\theta
        \le C s^{-\theta/2}(\log s)^{\theta/2}.
\end{equation}
Because $A_0^2<\theta$,
\[
        s^{-\theta/2}(\log s)^{\theta/2}=o\bigl(s^{-A_0^2/2}(\log s)^{-1/2}\bigr).
\]
Combining \eqref{eq:a2-margin-decomp}--\eqref{eq:gap-margin-overshoot-small} gives
\begin{equation}
        a_2-\tau_s(c)
        \ge c s^{-A_0^2/2}(\log s)^{-1/2}
\end{equation}
for all sufficiently large $s$.

It remains to check that the opposite edge $b_1$ is not close.  On the same side,
\begin{equation}\label{eq:b1-margin-decomp}
        \tau_s(c)-b_1
        \ge G\Prob(N\ge c)-\E[(b_1-H_{s,c})_+].
\end{equation}
Since $x_s(c)\le0$, the event $\{N\ge c\}$ has probability bounded below by a positive constant, uniformly in the present range:
\[
        \frac{c-\mu_s}{\sigma\sqrt{s}}
        =x_s(c)+O(s^{-1/2})\le O(s^{-1/2}),
\]
so Berry--Esseen gives $\Prob(N\ge c)\ge c_1>0$ for all large $s$.  The second term in \eqref{eq:b1-margin-decomp} is $o(1)$ by Lemma~\ref{lem:gapped-mean-overshoot}.  Thus $\tau_s(c)-b_1\ge Gc_1/2$ for large $s$, which is stronger than \eqref{eq:gapped-gap-margin-pub}.  The case $x_s(c)\ge0$ is identical after interchanging the two support intervals and the two events $\{N\ge c\}$ and $\{N<c\}$.
\end{proof}

\begin{lemma}[Activity of neighboring states]\label{lem:gapped-neighbor-active}
If at least one of
\[
        B_s(c),\qquad B_{s-1}(c),\qquad B_{s-1}(c-1)
\]
is nonzero, then for all sufficiently large $s$,
\begin{equation}\label{eq:neighbor-active-strong}
        |x_s(c)|\le \alpha_2\sqrt{\log s}+1.
\end{equation}
In particular,
\begin{equation}\label{eq:neighbor-active-bound}
        |x_s(c)|\le A \sqrt{\log s}.
\end{equation}
\end{lemma}

\begin{proof}
If $B_s(c)\ne0$, then \eqref{eq:neighbor-active-strong} follows from the support of $G_s$.  Suppose $B_{s-1}(c)\ne0$.  Then $|x_{s-1}(c)|\le\alpha_2\sqrt{\log(s-1)}$.  By Lemma~\ref{lem:predecessor-coordinates},
\[
        x_{s-1}(c)=x_s(c)+O(s^{-1/2})+O(|x_s(c)|/s).
\]
Since $\sqrt{\log(s-1)}=\sqrt{\log s}+o(1)$ and $\sqrt{\log s}\to\infty$, rearranging gives $|x_s(c)|\le\alpha_2\sqrt{\log s}+1$ for all large $s$.  The argument for $B_{s-1}(c-1)\ne0$ is the same.  Finally, because $\alpha_2<A$, the stronger bound \eqref{eq:neighbor-active-strong} implies \eqref{eq:neighbor-active-bound} for all sufficiently large $s$.
\end{proof}

\begin{lemma}[First difference and gap margin]\label{lem:gapped-first-diff-gap}
For the certificate \eqref{eq:gapped-B-def-pub}, fix any $\eta_0<\infty$.  There are $C_D<\infty$ and $s_D<\infty$, uniform for $0<\eta\le\eta_0$, such that whenever $s\ge s_D$ and a neighboring state is active,
\begin{equation}\label{eq:gapped-D-bound-pub}
        |D_s(c)|\le C_D\eta s^{-\theta/2}(\log s)^\theta.
\end{equation}
Moreover, for all sufficiently large $s$,
\begin{equation}\label{eq:perturbed-threshold-in-gap}
        \tau_s(c)-D_s(c)\in(b_1,a_2).
\end{equation}
Hence
\begin{equation}\label{eq:gapped-Delta-perturbed-zero}
        \Delta(\tau_s(c)-D_s(c))=0.
\end{equation}
\end{lemma}

\begin{proof}
If both predecessor values $B_{s-1}(c)$ and $B_{s-1}(c-1)$ vanish, then $D_s(c)=0$ and \eqref{eq:gapped-D-bound-pub} is immediate.  Otherwise at least one predecessor value is nonzero.  By Lemma~\ref{lem:gapped-neighbor-active}, the active predecessor coordinate is within the support band enlarged by one.  The two predecessor coordinates differ by $O(s^{-1/2})$, so the whole line segment between them lies in the active range enlarged by one for all large $s$.  Hence the mean-value argument in Corollary~\ref{cor:standard-first-difference} applies.  Using the global first-derivative bound in Lemma~\ref{lem:gapped-profile-pub},
\[
        |D_s(c)|
        \le C\eta s^\gamma s^{-1/2}(\log s)^\theta
        =C\eta s^{-\theta/2}(\log s)^\theta,
\]
which proves \eqref{eq:gapped-D-bound-pub}.  By Lemma~\ref{lem:gapped-neighbor-active}, the state lies in the range of Lemma~\ref{lem:gapped-gap-margin-pub} with $A_0=A$.  Therefore the offline marginal has gap margin at least $c_Gs^{-A^2/2}(\log s)^{-1/2}$.  Since $A^2<\theta$,
\[
        \frac{s^{-\theta/2}(\log s)^\theta}{s^{-A^2/2}(\log s)^{-1/2}}
        =s^{-(\theta-A^2)/2}(\log s)^{\theta+1/2}\longrightarrow0.
\]
Thus, after increasing $s_D$ if necessary, uniformly for $0<\eta\le\eta_0$ the perturbation $D_s(c)$ is smaller than half the distance from $\tau_s(c)$ to the two gap boundaries.  This proves \eqref{eq:perturbed-threshold-in-gap}.  Since $\Delta$ vanishes on the gap, \eqref{eq:gapped-Delta-perturbed-zero} follows.
\end{proof}

\subsection{Drift and verification of the Bellman certificate inequality}\label{subsec:gapped-drift-feasibility}

We bound the affine-transport drift of the certificate by applying the finite-difference expansion from Appendix~\ref{app:fd-standard}.  The two regions below isolate the harmonic cancellation that motivates the profile.

\begin{lemma}[Drift of the harmonic certificate]\label{lem:gapped-drift-pub}
There is a constant $C_M<\infty$ such that, whenever a neighboring state is active,
\begin{equation}\label{eq:gapped-drift-pub}
        M_s(c)
        \ge -C_M\eta s^{\gamma-1}(1+|x_s(c)|)^{1+\theta}.
\end{equation}
\end{lemma}

\begin{proof}
Let $x=x_s(c)$.  By Lemma~\ref{lem:gapped-neighbor-active}, $|x|\le \alpha_2\sqrt{\log s}+1\le A\sqrt{\log s}$ for all large $s$, so the standard-scale expansion of Lemma~\ref{lem:standard-drift} applies with $a_s=s^\gamma$ and profile $G_s$.  Since
\[
        (s-1)^\gamma-s^\gamma=-\gamma s^{\gamma-1}+O(s^{\gamma-2}),
\]
we can rewrite the expansion as
\begin{equation}\label{eq:gapped-drift-start}
        M_s(c)
        =\eta s^{\gamma-1}\Bigl[\frac12G_s''(x)+\frac12xG_s'(x)-\gamma G_s(x)\Bigr]
        +\eta s^\gamma\bigl(G_{s-1}(x)-G_s(x)\bigr)
        +\mathcal R_{s,c},
\end{equation}
where $\mathcal R_{s,c}$ is bounded by the error terms in \eqref{eq:standard-drift-error}, with $a_s=s^\gamma$, together with the additional lower-order term
\[
        C\eta s^{\gamma-2}|G_s(x)|.
\]
This additional term is the $O(s^{\gamma-2})G_s(x)$ remainder from expanding $(s-1)^\gamma-s^\gamma$.  We show that every term on the right of \eqref{eq:gapped-drift-start} is bounded below by the right side of \eqref{eq:gapped-drift-pub}.

First suppose $|x|\le\alpha_1\sqrt{\log s}-1$.  Then the unit neighborhood of $x$ lies in the core for large $s$, and $G_s=(\log s)^\theta H_\gamma$ there.  The leading operator term in \eqref{eq:gapped-drift-start} is exactly zero by Lemma~\ref{lem:Hgamma-properties}.  The time-variation bound \eqref{eq:gapped-profile-time-core} gives
\[
        \eta s^\gamma |G_{s-1}(x)-G_s(x)|
        \le C\eta s^{\gamma-1}(\log s)^{\theta-1}(1+|x|)^{2\gamma}.
\]
Because $\theta-1<0$ and $1+\theta>2\gamma$, this is at most
\[
        C\eta s^{\gamma-1}(1+|x|)^{1+\theta}.
\]
The error terms in \eqref{eq:standard-drift-error} are smaller.  Using \eqref{eq:gapped-profile-global-derivative}, \eqref{eq:gapped-profile-time-derivatives}, and $|x|\le O(\sqrt{\log s})$, their contribution divided by $\eta s^{\gamma-1}$ is bounded by a sum of terms of the form
\[
        \frac{(\log s)^{1+\theta}}{s},
        \qquad
        \frac{(\log s)^{1/2+\theta}}{\sqrt{s}},
        \qquad
        \frac{(\log s)^{(1+\theta)/2}}{s},
        \qquad
        \frac{(\log s)^{(2+\theta)/2}}{s}.
\]
The last two terms include the time-difference errors, the cross term added in Lemma~\ref{lem:standard-drift}, and the $O(\eta s^{\gamma-2}|G_s(x)|)$ contribution above.  All displayed terms are $o(1)$ and hence the corresponding contributions to $M_s(c)$ are bounded by $C\eta s^{\gamma-1}(1+|x|)^{1+\theta}$.  Thus \eqref{eq:gapped-drift-pub} holds in the core.

Now suppose $\alpha_1\sqrt{\log s}-1<|x|\le\alpha_2\sqrt{\log s}+1$.  Then $(1+|x|)^{1+\theta}\asymp (\log s)^{(1+\theta)/2}$.  The collar bounds in Lemma~\ref{lem:gapped-profile-pub} imply
\[
        \left|\frac12G_s''(x)+\frac12xG_s'(x)-\gamma G_s(x)\right|
        \le C(\log s)^{(1+\theta)/2}.
\]
The time-variation term is at most $C\eta s^{\gamma-1}(\log s)^{(1+\theta)/2}$ by \eqref{eq:gapped-profile-time-derivatives}.  The remainder terms in \eqref{eq:standard-drift-error}, including the cross term involving $G_{s-1}-G_s$, and the additional $O(\eta s^{\gamma-2}|G_s(x)|)$ term above are also bounded by $C\eta s^{\gamma-1}(\log s)^{(1+\theta)/2}$ because $|x|=O(\sqrt{\log s})$, $|(s-1)^\gamma/s^\gamma-1|=O(1/s)$, and the profile derivatives have at most polynomial-in-$\sqrt{\log s}$ size.  Since $(\log s)^{(1+\theta)/2}\asymp(1+|x|)^{1+\theta}$ in this region, \eqref{eq:gapped-drift-pub} follows.
\end{proof}

\begin{proposition}[Feasibility of the gapped harmonic certificate]\label{prop:gapped-feasible-pub}
There exist $s_0<\infty$ and $\eta_0>0$ such that, for every $\eta\in(0,\eta_0]$, the certificate \eqref{eq:gapped-B-def-pub} satisfies the Bellman certificate inequality \eqref{eq:common-bellman-residual} for every $s>s_0$ and $1\le c<s$.
\end{proposition}

\begin{proof}
Fix $s$ and $c$.  If all three neighboring certificate values
\[
        B_s(c),\qquad B_{s-1}(c),\qquad B_{s-1}(c-1)
\]
are zero, then $M_s(c)=0$ and $D_s(c)=0$.  The Bellman certificate inequality becomes
\[
        \E\Delta(H_{s,c})-\Delta(\tau_s(c))\ge0,
\]
which is the offline Jensen slack from Lemma~\ref{lem:offline-slack-common}.

It remains to treat the case in which a neighboring state is active.  Lemma~\ref{lem:gapped-first-diff-gap} gives $\Delta(\tau_s(c)-D_s(c))=0$.  Therefore the source term in the Bellman certificate inequality is simply
\[
        S_s(c,D_s(c))=\E\Delta(H_{s,c}).
\]
By Lemmas~\ref{lem:gapped-neighbor-active} and \ref{lem:gapped-source-pub},
\[
        S_s(c,D_s(c))
        \ge c_Ss^{\gamma-1}(1+|x_s(c)|)^{1+\theta}.
\]
By Lemma~\ref{lem:gapped-drift-pub},
\[
        M_s(c)
        \ge -C_M\eta s^{\gamma-1}(1+|x_s(c)|)^{1+\theta}.
\]
Choose $\eta_0\le c_S/(2C_M)$ and choose $s_0$ large enough that all preceding lemmas apply for $s>s_0$.  Then
\[
        M_s(c)+S_s(c,D_s(c))
        \ge \frac{c_S}{2}s^{\gamma-1}(1+|x_s(c)|)^{1+\theta}\ge0.
\]
This is exactly \eqref{eq:common-bellman-residual}.
\end{proof}

\subsection{Shifted capacity lower bound}\label{subsec:gapped-positive-theorems}

\begin{theorem}[Gapped lower bound at shifted capacity, $\beta>0$]\label{thm:gapped-shift-positive}
Assume the two-support gapped model and Assumption~\ref{ass:model-gapped} with $\beta>0$.  There is $\bar\alpha>0$ such that, for every fixed $\alpha\in(0,\bar\alpha)$, there exists $c_\alpha>0$ for which
\begin{equation}\label{eq:gapped-shift-positive-result}
        B_T^\star\!\left(\left\lfloor qT+\alpha\sigma\sqrt{T\log T}\right\rfloor\right)
        \ge c_\alpha T^\gamma(\log T)^{(\beta+2)/(2(\beta+1))}
\end{equation}
for all sufficiently large $T$.
\end{theorem}

\begin{proof}
Let $s_0$ and $\eta_0$ be as in Proposition~\ref{prop:gapped-feasible-pub}.  By Corollary~\ref{cor:base-scale-framework}, after possibly reducing $\eta\in(0,\eta_0]$, the certificate satisfies the base condition at time $s_0$.

Take $\bar\alpha:=\alpha_1/2$ and fix $\alpha\in(0,\bar\alpha)$.  Let
\[
        k_T^{(\alpha)}:=\left\lfloor qT+\alpha\sigma\sqrt{T\log T}\right\rfloor.
\]
Proposition~\ref{prop:gapped-feasible-pub} verifies the Bellman certificate inequality for all later times, so $B$ is feasible for $\mathsf P(T,k_T^{(\alpha)},s_0)$ for all sufficiently large $T$.  Proposition~\ref{prop:common-comparison} gives
\[
        B_T^\star(k_T^{(\alpha)})=\Reg(T,k_T^{(\alpha)};F)\ge B_T(k_T^{(\alpha)}).
\]
Then
\[
        x_T(k_T^{(\alpha)})=\alpha\sqrt{\log T}+o(1).
\]
Since $\alpha<\alpha_1$, the cutoff is equal to one for all large $T$.  By the large-$x$ asymptotic in Lemma~\ref{lem:Hgamma-properties},
\[
        H_\gamma(\alpha\sqrt{\log T}+o(1))\ge c_\alpha(\log T)^\gamma
\]
for all large $T$; the $o(1)$ shift is negligible because $\alpha>0$ is fixed.  Absorbing the fixed certificate multiplier $\eta$ into $c_\alpha$, we obtain
\[
        B_T(k_T^{(\alpha)})
        \ge c_\alpha T^\gamma(\log T)^{\theta+\gamma}.
\]
Using $\theta+\gamma=(1+\theta)/2$, this becomes
\[
        B_T(k_T^{(\alpha)})
        \ge c_\alpha T^\gamma(\log T)^{(1+\theta)/2}.
\]
Finally,
\[
        \frac{1+\theta}{2}=\frac{\beta+2}{2(\beta+1)}.
\]
Combining this with the certified lower bound above gives \eqref{eq:gapped-shift-positive-result}.
\end{proof}

\section{Moderate-scale finite-difference details for the bounded-density case}\label{app:g0-asymp}

This appendix records the finite-difference expansions behind Lemma~\ref{lem:g0-fdiff}.  They are the asymptotic calculations used in the two-uniform proof.  Throughout the appendix, write $d=d_s(c)$ and $z=z_s(c)$, with $d_s$ and $z_s$ defined in \eqref{eq:coords-model}--\eqref{eq:z-def-model}.  The profile $\varphi_s$ is the slow-tail profile of Lemma~\ref{lem:g0-profile}.

\subsection{Scalar expansions}

\begin{lemma}[Scalar $s$-asymptotics]\label{lem:g0-s-asymp-detail}
As $s\to\infty$,
\begin{align}
        (\log(s-1))^2-(\log s)^2
        &= -\frac{2\log s}{s}+O\left(\frac{\log s}{s^2}\right),
        \label{eq:g0-app-ell-sq-diff}\\
        \frac{\sigma\sqrt{s\log s}}{\sigma\sqrt{(s-1)\log(s-1)}}
        &=1+\frac1{2s}+\frac1{2s\log s}+O(s^{-2}),
        \label{eq:g0-app-r-ratio}\\
        \frac{(\log(s-1))^2}{(\log s)^2}
        &=1-\frac{2}{s\log s}+O\left(\frac1{s^2\log s}\right).
        \label{eq:g0-app-ell-ratio}
\end{align}
\end{lemma}

\begin{proof}
The identity
\[
        \log(s-1)=\log s+\log(1-1/s)=\log s-1/s+O(s^{-2})
\]
implies \eqref{eq:g0-app-ell-sq-diff} by squaring.  Since the square of the normalizer is $\sigma^2s\log s$,
\[
        \frac{\sigma^2s\log s}{\sigma^2(s-1)\log(s-1)}
        =\frac{s}{s-1}\frac{\log s}{\log(s-1)}
        =\left(1+\frac1s+O(s^{-2})\right)
          \left(1+\frac1{s\log s}+O((s^2\log s)^{-1})\right),
\]
and taking square roots gives \eqref{eq:g0-app-r-ratio}.  Finally,
\[
        \log(s-1)/\log s=1-1/(s\log s)+O((s^2\log s)^{-1}),
\]
and squaring gives \eqref{eq:g0-app-ell-ratio}.
\end{proof}

\subsection{Sign-stable predecessor representation}

\begin{lemma}[Sign-stable representation]\label{lem:g0-sign-stable-detail}
Fix $\alpha_0>0$.  For all sufficiently large $s$, if $z_s(c)\ge\alpha_0$, then the two predecessor deviations $d_s(c)+q$ and $d_s(c)-(1-q)$ have the same sign as $d_s(c)$.  Let $\xi$ be the two-point random variable
\[
        \xi=q\,\sgn(d_s(c)) \quad\text{with probability }1-q,
        \qquad
        \xi=-(1-q)\,\sgn(d_s(c)) \quad\text{with probability }q.
\]
Then $\E\xi=0$, $\E\xi^2=\sigma^2$, $|\xi|\le1$, and, with
\[
        Z_s:=\frac{|d_s(c)|+\xi}{\sigma\sqrt{(s-1)\log(s-1)}},
\]
one has
\begin{equation}\label{eq:g0-app-convex-exp}
        (1-q)\varphi_{s-1}(z_{s-1}(c))+q\varphi_{s-1}(z_{s-1}(c-1))
        =\E \varphi_{s-1}(Z_s).
\end{equation}
\end{lemma}

\begin{proof}
If $z_s(c)\ge\alpha_0$, then $|d_s(c)|=z_s(c)\sigma\sqrt{s\log s}\to\infty$.  Thus, for all large $s$, both predecessor deviations have the sign of $d_s(c)$, and
\[
        |d_s(c)+q|=|d_s(c)|+q\sgn(d_s(c)),
        \qquad
        |d_s(c)-(1-q)|=|d_s(c)|-(1-q)\sgn(d_s(c)).
\]
The moment identities for $\xi$ are immediate.  The two values taken by $Z_s$ are precisely $z_{s-1}(c)$ and $z_{s-1}(c-1)$ with probabilities $1-q$ and $q$, respectively, proving \eqref{eq:g0-app-convex-exp}.
\end{proof}

\begin{lemma}[Moments of the predecessor shift]\label{lem:g0-dZ-detail}
Let $Z_s$ be as in Lemma~\ref{lem:g0-sign-stable-detail}, set $\delta Z:=Z_s-z$, where $z=z_s(c)$, and assume $z\le C_0\sqrt{s/(\log s)}$ for a fixed constant $C_0$.  Then
\begin{equation}\label{eq:g0-app-Z-expansion}
        Z_s=z+\frac{\xi}{\sigma\sqrt{s\log s}}+\frac{z}{2s}+R_s,
        \qquad
        |R_s|\le C\frac{1+z}{s\log s},
\end{equation}
and
\begin{align}
        \E\delta Z
        &=\frac{z}{2s}+O\left(\frac{1+z}{s\log s}\right),
        \label{eq:g0-app-dZ1}\\
        \E(\delta Z)^2
        &=\frac1{s\log s}+O\left(\frac{1+z^2}{s^2}+\frac{1+z}{(s\log s)^{3/2}}\right),
        \label{eq:g0-app-dZ2}\\
        \E|\delta Z|^3
        &=O((s\log s)^{-3/2}).
        \label{eq:g0-app-dZ3}
\end{align}
\end{lemma}

\begin{proof}
By Lemma~\ref{lem:g0-s-asymp-detail},
\[
        Z_s=\left(z+\frac{\xi}{\sigma\sqrt{s\log s}}\right)\frac{\sigma\sqrt{s\log s}}{\sigma\sqrt{(s-1)\log(s-1)}}
        =\left(z+\frac{\xi}{\sigma\sqrt{s\log s}}\right)
          \left(1+\frac1{2s}+\frac1{2s\log s}+O(s^{-2})\right).
\]
This gives \eqref{eq:g0-app-Z-expansion}.  The omitted term involving $\xi/(2s\sigma\sqrt{s\log s})$ is $O((s^{3/2}\sqrt{\log s})^{-1})$, which is dominated by $(1+z)/(s\log s)$ on the stated range.  Taking expectations and using $\E\xi=0$ gives \eqref{eq:g0-app-dZ1}.  Squaring \eqref{eq:g0-app-Z-expansion}, using $\E\xi^2=\sigma^2$ and the identity $(\sigma\sqrt{s\log s})^2=\sigma^2s\log s$, gives the leading term in \eqref{eq:g0-app-dZ2}.  The deterministic square of the remainder contributes $O((1+z)^2/(s^2(\log s)^2))$, and the cross-term between $\xi/(\sigma\sqrt{s\log s})$ and $R_s$ contributes $O((1+z)/(s\log s)^{3/2})$ by Cauchy--Schwarz; both are covered by the displayed error bound.  The third-moment bound follows from $|a+b+c|^3\le 27(|a|^3+|b|^3+|c|^3)$ and the assumed range $z\le C_0\sqrt{s/(\log s)}$.
\end{proof}

\subsection{Bounded-range Taylor expansion}

\begin{lemma}[Taylor expansion on a bounded $z$-range]\label{lem:g0-taylor-detail}
Let $Z_s$ be as above and suppose $z\in[\alpha_0,\alpha_5]$.  Then
\begin{equation}\label{eq:g0-app-taylor}
        (\log(s-1))^2\E\bigl[\varphi_{s-1}(Z_s)-\varphi_{s-1}(z)\bigr]
        =(\log s)^2\varphi_s'(z)\frac{z}{2s}
         +\frac12\log s\varphi_s''(z)\frac1s
         +R_T,
\end{equation}
where
\begin{equation}
        |R_T|\le C_\lambda \varphi_s(z)\frac{\log s}{s}.
\end{equation}
\end{lemma}

\begin{proof}
Taylor's theorem gives
\[
        \varphi_{s-1}(Z_s)-\varphi_{s-1}(z)
        =\varphi_{s-1}'(z)\delta Z+\frac12\varphi_{s-1}''(z)(\delta Z)^2+\Theta_s,
\]
with $|\Theta_s|\le C_\lambda |\delta Z|^3$ uniformly on $[\alpha_0,\alpha_5]$, by Lemma~\ref{lem:g0-profile}.  Taking expectations and using \eqref{eq:g0-app-dZ1}--\eqref{eq:g0-app-dZ3},
\[
        \E[\varphi_{s-1}(Z_s)-\varphi_{s-1}(z)]
        =\varphi_{s-1}'(z)\frac{z}{2s}
          +\frac12\varphi_{s-1}''(z)\frac1{s\log s}
          +R_1,
\]
where $(\log(s-1))^2|R_1|\le C_\lambda \varphi_s(z)\log s/s$.  The time-variation bounds \eqref{eq:g0-profile-der-time} allow $\varphi_{s-1}'$ and $\varphi_{s-1}''$ to be replaced by $\varphi_s'$ and $\varphi_s''$; after multiplication by $(\log(s-1))^2$, the resulting contribution is absorbed by the same remainder.  Multiplying by $(\log(s-1))^2=(\log s)^2(1+O((s\log s)^{-1}))$ proves \eqref{eq:g0-app-taylor}.
\end{proof}

\subsection{Ratio expansion for large deviations}

\begin{lemma}[Ratio expansion for $z\ge\alpha_5$]\label{lem:g0-tail-ratio-detail}
Let $Z_s$ be as in Lemma~\ref{lem:g0-sign-stable-detail}.  Suppose
\[
        z\in[\alpha_5,C_0\sqrt{s/(\log s)}]
\]
for a fixed $C_0<\infty$.  Then $Z_s\ge\alpha_4$ almost surely for all sufficiently large $s$, and
\begin{equation}
        \E\left[\frac{\varphi_{s-1}(Z_s)}{\varphi_s(z)}\right]
        =1-\frac{\lambda z}{2s\log s}
          +O_\lambda\left(\frac{1+z}{s(\log s)^2}+\frac1{s(\log s)^3}\right).
\end{equation}
\end{lemma}

\begin{proof}
Since $z\ge\alpha_5>\alpha_4$ and $|Z_s-z|=o(1)$ uniformly in the displayed range, all arguments lie in $[\alpha_4,\infty)$ for large $s$.  In that region,
\[
        \varphi_s(w)=\varphi_s(\alpha_4)\exp\{-\lambda(w-\alpha_4)/\log s\}.
\]
Thus
\[
\frac{\varphi_{s-1}(Z_s)}{\varphi_s(z)}
=\frac{\varphi_{s-1}(\alpha_4)}{\varphi_s(\alpha_4)}
 \exp\left\{-\frac{\lambda}{\log s}(Z_s-z)
          -\lambda\left(\frac1{\log(s-1)}-\frac1{\log s}\right)(Z_s-\alpha_4)\right\}.
\]
By the time-variation bound in Lemma~\ref{lem:g0-profile}, the prefactor is $1+O_\lambda((s(\log s)^2)^{-1})$, and
\[
        \frac1{\log(s-1)}-\frac1{\log s}=O((s(\log s)^2)^{-1}).
\]
Write the exponent as $A_s+B_s$, where
\[
        A_s:=-\frac{\lambda}{\log s}(Z_s-z),
        \qquad
        B_s:=-\lambda\left(\frac1{\log(s-1)}-\frac1{\log s}\right)(Z_s-\alpha_4).
\]
The moment estimates in Lemma~\ref{lem:g0-dZ-detail} imply
\[
        \E A_s=-\frac{\lambda z}{2s\log s}
          +O_\lambda\left(\frac{1+z}{s(\log s)^2}\right),
        \qquad
        \E a_s(c)^2=O_\lambda\left(\frac1{s(\log s)^3}
           +\frac{z^2}{s^2(\log s)^2}\right).
\]
In the displayed range $z\le C_0\sqrt{s/(\log s)}$, the final term is
$O((s(\log s)^3)^{-1})$.  Also
\[
        |B_s|\le C_\lambda\frac{1+z+|Z_s-z|}{s(\log s)^2},
\]
so $\E|B_s|=O_\lambda((1+z)/(s(\log s)^2))$.  The third-order remainder in the exponential is bounded by
$O_\lambda(\E|A_s|^3+\E|B_s|^2+\E|A_sB_s|)$, which is absorbed by
$O_\lambda((1+z)/(s(\log s)^2)+1/(s(\log s)^3))$.  Hence
\[
        \E e^{A_s+B_s}
        =1-\frac{\lambda z}{2s\log s}
        +O_\lambda\left(\frac{1+z}{s(\log s)^2}+\frac1{s(\log s)^3}\right).
\]
Multiplying by the prefactor $1+O_\lambda((s(\log s)^2)^{-1})$ proves the lemma.
\end{proof}

\end{document}